\documentclass[sigconf, nonacm]{acmart}

\usepackage{subfigure}
\usepackage[ruled,linesnumbered, vlined,titlenotnumbered]{algorithm2e}
\usepackage{caption}
\usepackage{amsthm}
\usepackage{amsmath} 
\usepackage{mathtools} 
\usepackage{enumitem}
\usepackage{multirow} 
\usepackage{balance}  

\usepackage{xcolor}
\usepackage{tabularx}

\newcommand\vldbdoi{XX.XX/XXX.XX}

\newtheoremstyle{common}
 {} 
 {} 
 {\rm} 
 {0.7em} 
 {\bfseries} 
 {.} 
 {.5em} 
 {} 
\theoremstyle{common}

\newtheorem{exam}{E{\scriptsize XAMPLE}}
\newtheorem{defn}{D{\scriptsize EFINITION}}
\newtheorem{lmma}{L{\scriptsize EMMA}}

\newtheorem{coro}{C{\scriptsize OROLLARY}}
\newtheorem{prop}{P{\scriptsize ROPOSITION}}
\newtheorem{cond}{C{\scriptsize ONDITION}}
\newtheorem{claim}{C{\scriptsize LAIM}}

\newcommand{\COMMENTOUT}[1]{}

\begin{document}
\title{Nass: A New Approach to Graph Similarity Search}

\author{Jongik Kim}
\affiliation{%
  \institution{Jeonbuk National University}
  \city{Jeonju}
  \country{Republic of Korea}
}
\email{jongik@jbnu.ac.kr}

\begin{abstract}
  In this paper, we study the problem of graph similarity search with graph
  edit distance (GED) constraints. Due to the NP-hardness of GED computation,
  existing solutions to this problem adopt the filtering-and-verification
  framework with a main focus on the filtering phase to generate a small number
  of candidate graphs. However, they have a limitation that the number of
  candidates grows extremely rapidly as a GED threshold increases. To address
  the limitation, we propose a new approach that utilizes GED computation
  results in generating candidate graphs. The main idea is that whenever we
  identify a result graph of the query, we immediately regenerate candidate
  graphs using a subset of pre-computed graphs similar to the identified result
  graph.  To speed up GED computation, we also develop a novel GED computation
  algorithm.  The proposed algorithm reduces the search space for GED
  computation by utilizing a series of filtering techniques, which have been
  used to generate candidates in existing solutions. Experimental results on
  real datasets demonstrate the proposed approach significantly outperforms the
  state-of-the art techniques.
\end{abstract}

\maketitle


\makeatletter
\def\thickhline{%
  \noalign{\ifnum0=`}\fi\hrule \@height \thickarrayrulewidth \futurelet
   \reserved@a\@xthickhline}
\def\@xthickhline{\ifx\reserved@a\thickhline
               \vskip\doublerulesep
               \vskip-\thickarrayrulewidth
             \fi
      \ifnum0=`{\fi}}
\makeatother

\newlength{\thickarrayrulewidth}
\setlength{\thickarrayrulewidth}{2\arrayrulewidth}

\section{Introduction}
\label{sec:intro}

Complex and interconnected data, which are represented by graph models, are
used in a wide range of applications such as business process management,
pattern recognition, drug design, program analysis, and
cheminformatics~\cite{PARS, BRANCH, GSIM2, SEGOS, KAT, MCS1, MLINDEX}.  Finding
graphs similar to a given query is a fundamental operation required in these
applications, because inconsistency, natural noises, and different
representations are unavoidable in real-world graph data.

To quantify the similarity between graphs, various similarity measures have
been developed, such as maximum common subgraphs~\cite{MCS1,MCS2}, missing
edges and features~\cite{MISSEDGE2,MISSEDGE3}, and graph
alignment~\cite{GRAPHALIGN}. Among them, the most commonly used measure in
similarity search studies is graph edit distance (GED)~\cite{GEDORG,
  GED-SURVEY, GED1, GED-1, GED-2, GED-3, GED-4}. The GED between two graphs is
the minimum number of edit operations to transform one graph into the other,
where an edit operation is insertion, deletion, or relabeling of a single
vertex or edge.
GED can capture the structural difference between graphs, and it can be applied
to many types of graphs~\cite{STAR, PARS}.

The graph similarity search problem studied in this paper is to retrieve all
graphs in a collection of graphs whose GED to a query is within a given
threshold. The NP-hardness of GED computation~\cite{STAR} makes this problem
challenging, and there has been a rich literature in developing efficient graph
similarity search techniques. Existing solutions to the problem have been
developed under a filtering-and-verification framework, where candidate graphs
are generated using various filtering techniques, and each of the candidates is
verified by GED computation. In the filtering-and-verification paradigm, it is
crucial to generate a set of candidates as small as possible because the
performance of verification relies on the number of candidates. Therefore, the
majority of efforts have been aimed at developing candidate generation
techniques.

To generate candidates, existing techniques establish a filtering condition
between dissimilar graphs by utilizing features of graphs, i.e., substructures
of graphs. To eliminate the overhead of extracting features from data graphs
and to quickly generate candidates, most techniques build an offline index on
features of data graphs. For example, {\sf c-star}~\cite{STAR} and {\sf
  $k$-AT}~\cite{KAT} build an index on tree-structured features extracted from
data graphs. {\sf Branch}~\cite{BRANCH, MIXED} and {\sf
  GSimSearch}~\cite{GSIM1, GSIM2} index branch and path-based $q$-gram
features, respectively. {\sf Pars}~\cite{PARS, PARS2} and {\sf
  MLIndex}~\cite{MLINDEX} use partitions of graphs as features to be
indexed. In contrast to others, a recent technique {\sf Inves}~\cite{INVES}
introduces an online-partitioning algorithm to make use of a partition-based
filter in the verification phase. Its partition-based filter plays a role of
screening each candidate to reduce the number of candidates passed to GED
computation. Hence, this candidate refinement step of {\sf Inves} can be also
considered as a part of candidate generation.

\begin{table}[htbp]
\centering
\caption{ Number of candidates vs. number of results of existing algorithms
on the AIDS dataset (avg. of 100 queries)}
  \begin{tabular}{c|c|c|c|c|c|c} \thickhline
    \multirow{2}{*}{$\tau$} & \multicolumn{5}{c|}{\small Candidates} & \multirow{2}{*}{\small Results} \\\cline{2-6}
                            & {\sf \small LF} & {\sf \small GSimSearch} & {\sf \small MLIndex} & {\sf \small Pars} & {\sf \small Inves} & \\ \hline\hline
       1 & 13   & 7.6  & 2.7  & 1.3  & 0.5  & 0.22 \\
       2 & 78   & 63   & 37   & 14   & 4.6  & 0.63 \\
       3 & 285  & 273  & 226  & 139  & 59   & 1.26 \\  
       4 & 738  & 736  & 691  & 562  & 335  & 2.70 \\
       5 & 1488 & 1487 & 1459 & 1346 & 1035 & 4.91 \\ 
       6 & 2514 & 2514 & 2493 & 2462 & 2129 & 9.09 \\  
       7 & 3780 & 3780 & 3757 & 3751 & 3501 & 18.45 \\ \thickhline
  \end{tabular}
\label{tbl:expr-intro-cand}  
\end{table}

An inherent limitation of existing feature-based filtering techniques is that
the filtering effect sharply decreases as a GED threshold increases.
Table~\ref{tbl:expr-intro-cand} shows the average number of candidates and that
of results of existing filtering algorithms on a real dataset AIDS for 100
queries (see Section~\ref{sec:expr} for details of the dataset and queries). In
the table, $\tau$ denotes a GED threshold, and {\sf LF} is a basic filter that
utilizes the difference between label multisets of graphs.
{\sf LF} gives an upper bound on the number of candidates. From the results in
the table, we observe that the number of candidates grows significantly
faster than that of results until it almost reaches an obvious upper bound,
i.e., the number of candidates generated by {\sf LF}.

Motivated by the observation, we propose in Section~\ref{sec:framework} a
fundamentally different filtering approach that makes use of GED computation
results in generating candidate graphs. The main idea is that if we identify a
graph $r$ as a result of the query, we immediately (re)generate candidate
graphs using a subset of pre-computed graphs similar to $r$. Unlike existing
techniques that strictly separate the filtering phase from the verification
phase, our approach makes the filtering phase interact with the verification
phase. That is, candidate graphs are continuously regenerated whenever a result
of the query is identified in the verification phase.

To compute the GED between a pair of graphs, existing
solutions~\cite{INVES,GED,PARS,GSIM1,GSIM2,PARS2,MIXED,STAR,MLINDEX} perform a
best-first search in the space of all possible vertex mappings between the
pair, which are organized into a prefix tree. An intermediate node of the
prefix tree represents a partial vertex mapping. At each tree node, a GED lower
bound of the corresponding partial vertex mapping is computed, and the subtree
rooted by the node is pruned if the lower bound is greater than the
threshold. In computing a GED lower bound of a partial mapping, all existing
solutions exploit label set differences of unmapped vertices and
edges. However, a label set-based lower bound tends to be very loose since it
cannot reflect structural differences. As a consequence, existing solutions
suffer from a huge search space.

To reduce the search space, we formulate the GED computation problem as a
repetition of filtering dissimilar subgraphs. Although the existing filtering
techniques exhibit poor performance as a threshold increases, we show in
Section~\ref{sec:ged} that they can be effectively used to filter out
dissimilar subgraphs during GED computation. Based on the observation, we
develop a novel and efficient GED computation algorithm that integrates
alternative filtering techniques, which have been used for generating
candidates.

In summary, the following are the contributions of this paper.

\begin{itemize}[leftmargin=*]
\item We propose a new approach to graph similarity search that exploits
  GED computation results in generating candidates and utilizes a series of
  filtering techniques in GED computation.
\item
  We show that candidate graphs can be dynamically regenerated while verifying
  candidate graphs, and propose a novel graph similarity search algorithm
  based on the candidate regeneration method.
\item We develop an efficient GED computation algorithm that substantially
  reduces the search space by utilizing a series filtering techniques. We
  judiciously select filters for GED computation, and carefully apply the
  selected filters to efficiently prune dissimilar subgraphs.
\item We integrate the proposed techniques into a new search framework named
  {\sf Nass} (\underline{\bf n}ew \underline{\bf a}ppro\-ach to graph
  \underline{\bf s}imilarity \underline{\bf s}earch), and implement the
  framework. We conduct extensive experiments on real datasets and show that
  {\sf Nass} outperforms existing techniques by orders of magnitude.
\end{itemize}

The rest of the paper is organized as follows: Section~\ref{sec:preli} presents
preliminaries and related work.
Section~\ref{sec:framework} proposes a new method to generate candidate graphs
via GED verification, and presents our search framework.
Section~\ref{sec:ged} develops an
efficient GED computation algorithm that takes advantage of a series filtering
techniques.
Section~\ref{sec:impl} presents implementation issues, and
Section~\ref{sec:expr} reports experimental
results. Section~\ref{sec:conclusion} concludes the paper.

\section{Preliminaries and Related Work}
\label{sec:preli}

\subsection{Problem Formulation}
In this paper, we focus on undirected and labeled simple graphs, though the
proposed technique can be easily extended other types of graphs.  An undirected
and labeled simple graph $g$ is a triple $(V(g),\ E(g),\ l)$, where $V(g)$ is
a set of vertices, $E(g) \subseteq V(g) \times V(g)$ is a set of edges, and
$l:V(g) \cup (V(g) \times V(g)) \rightarrow \Sigma$ is a labeling function that
maps vertices and edges to labels, where $\Sigma$ is the label set of vertices
and edges. $l(v)$ and $l(u, v)$ respectively denote the label of a vertex
$v$ and the label of an edge $(u, v)$.
If there is no edge between $u$ and $v$, $l(u,v)$ returns a unique value
$\lambda$ distinguished from all other labels. We also define a blank vertex
$\varepsilon$ such that $l(\varepsilon) = l(\varepsilon, v) = l(u,
\varepsilon) = \lambda$.
There are no self-loops nor more than one edge between two vertices.  For
simplicity, in the rest of the paper, we use graph to denote undirected and
labeled simple graph.

To measure the similarity between a pair of graphs, we use graph edit distance defined
in Definition~\ref{def:ged}.

\begin{defn}[{\bf Graph edit distance}]
  \label{def:ged}
The graph edit distance (GED) between two graphs $g_1$ and $g_2$, which is
denoted by $\mathsf{ged}(g_1, g_2)$, is the minimum number of edit operations
that transform $g_1$ into $g_2$, where an edit operation is one of the following:
\begin{enumerate}[leftmargin=*]
\label{def:ged}
\item insertion of an isolated labeled vertex
\item deletion of an isolated labeled vertex 
\item substitution of the label (i.e., relabeling) of a vertex
\item insertion of a labeled edge
\item deletion of a labeled edge
\item substitution of the label (i.e., relabeling) of an edge.
\end{enumerate}
\end{defn}

\begin{figure}[htbp]
  \centering \includegraphics[height=3.3cm]{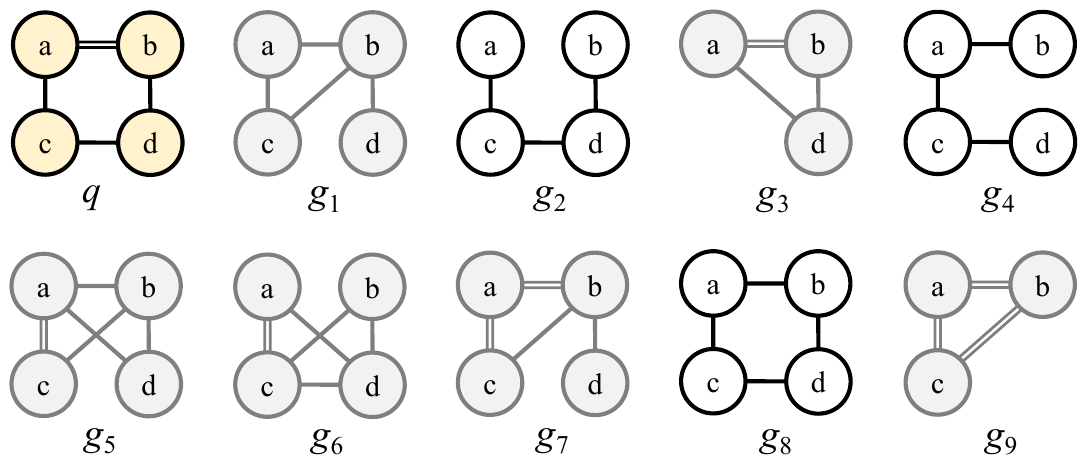}
  \caption{Example query and data graphs}
\label{fig:graph}
\end{figure}

\begin{exam}
\label{ex:ged}
  Consider two graphs $q$ and $g_1$ in Figure~\ref{fig:graph}, where the solid
  and hollow lines represent edge labels. To transform $q$ into $g_1$, we can
  perform the following three edit operations on $q$: substitution of the label
  of the edge between $a$ and $b$ (from a hollow edge to a solid edge),
  insertion of a solid edge between $b$ and $c$, and deletion of the edge
  between $c$ and $d$. Therefore, ${\sf ged}(q, g_1) = 3$.
\end{exam}

\begin{lmma}
\label{lm:metric}  
  GED defined in Definition~\ref{def:ged} is metric~\cite{GED-SURVEY}, and the following
  properties hold on GED.
  \begin{itemize}[leftmargin=*]
  \item $\forall g_1, g_2\ \ \mathsf{ged}(g_1, g_2) \geq 0$.
  \item $\forall g_1, g_2\ (g_1 = g_2\ \iff\ \mathsf{ged}(g_1, g_2) = 0)$.
  \item $\forall g_1, g_2\ \ \mathsf{ged}(g_1, g_2) = \mathsf{ged}(g_2, g_1)$.
  \item $\forall g_1, g_2, g_3\ \ \mathsf{ged}(g_1, g_2) \leq \mathsf{ged}(g_1, g_3) + \mathsf{ged}(g_2, g_3)$.
  \end{itemize}
\end{lmma}

We formulate the problem of graph similarity search in a graph database, as
follows.
  
\begin{defn}[{\bf Graph similarity search problem}]
For a graph database $\mathcal{D}$ and a query graph $q$ with a GED threshold
$\tau$, the problem of graph similarity search is to find a result set, denoted by
$\mathcal{R}(q, \tau)$, containing all data graphs $g \in \mathcal{D}$ such that
${\sf ged}(q, g) \leq \tau$.
%
\end{defn}

\colorlet{dred}{black!20!red}
\begin{exam}
\label{ex:simsearch}
  For the graphs in Figure~\ref{fig:graph}, consider a graph database
  $\mathcal{D} = \{g_1, g_2, \ldots, g_9\}$, and a query graph $q$ with a GED
  threshold $\tau = 2$. The following table shows the GED between $q$ and $g_i$.
  \begin{center}
  \begin{tabular}{p{1.5cm}|c|c|c|c|c|c|c|c|c c} \thickhline
        & $g_1$ & $g_2$ & $g_3$ & $g_4$ & $g_5$ & $g_6$ & $g_7$ & $g_8$ & $g_9$ & \\ \hline
    \ {\sf ged}$(q, g_i)$ &  $3$  & \textcolor{dred}{\bf 1}&  $4$ &\textcolor{dred}{\bf 2}& $4$  &  $4$  &  $3$ &\textcolor{dred}{\bf 1} & $5$ & \\ \thickhline
  \end{tabular}
  \end{center}
\vspace*{0cm}
Hence, the graph similarity search returns $\mathcal{R}(q, 2) = \{g_2, g_4, g_8\}$.
\end{exam}


\subsection{GED Computation}
\label{sec:gedcomputation}

\newcommand*{\smapsto}{\mathbin{\scalebox{0.7}{\ensuremath{\mapsto}}}}

In this subsection, we provide a general description of existing GED
computation methods. A {\it vertex mapping} between two graphs $g_1$ and $g_2$
is a bijection of $V(g_1)$ onto $V(g_2)$\footnote{If $|V(g_1)| \neq |V(g_2)|$,
  we add $||V(g_1)| - |V(g_2)||$ copies of a blank vertex $\varepsilon$ into
  $V(g_1)$ or $V(g_2)$ to make $|V(g_1)| = |V(g_2)|$ based on the observation
  of \cite{MIXED}.}. A vertex mapping is represented by an ordered set of
mapped vertex pairs, where the order is imposed by a pre-defined ordering of
$V(g_2)$. Given a vertex mapping $m$, $g_1$ can be transformed into $g_2$ by
abiding by $m$ as follows. For each mapped vertex pair $u \smapsto v \in m$, we
make $u$ and $v$ identical in terms of the labels of the vertices and the
labels and connectivity of their adjacent edges. The number of edit operations
required in this transformation is called the {\it edit cost} of $m$, which is
formally stated in Definition~\ref{def:edcost}.

\begin{defn}[{\bf Edit cost}]
\label{def:edcost}  
  Let $u \smapsto v$ be the last mapped vertex pair in $m$, and $m' =
  m - \{u \smapsto v\}$. The edit cost of $m$ is defined as:
  \[{\sf ec}(m) = {\sf ec}(m') + {\sf d}[l(u), l(v)] +
  \sum_{u'\smapsto v' \in m'} {\sf d}[l(u, u'), l(v, v')],\]
  where ${\sf ec}(\emptyset) = 0$ and ${\sf d}[x, y]= \begin{cases} 0, & \text{if } x
    = y\\ 1, & \text{otherwise.} \end{cases}$
\end{defn}  

GED computation is a process to find a vertex mapping having a minimum edit
cost among all possible vertex mappings between $g_1$ and $g_2$.
To avoid redundant edit cost computation among vertex mappings that shares a
prefix, all possible vertex mappings can be organized into a prefix-sharing
tree, which is called a {\it search tree}.

\begin{figure}[htbp]
   \centering \includegraphics[height=1.7cm]{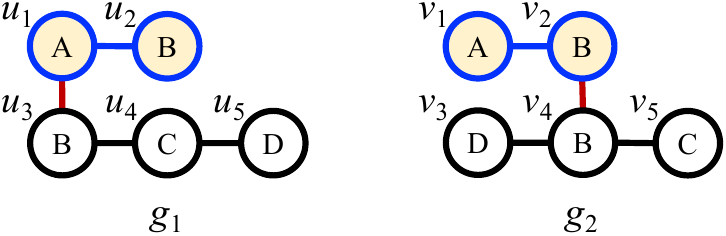}
   \caption{Two graphs $g_1$ and $g_2$}
\label{fig:mapping}
\end{figure}

\begin{figure}[t]
  \centering \includegraphics[height=3.4cm]{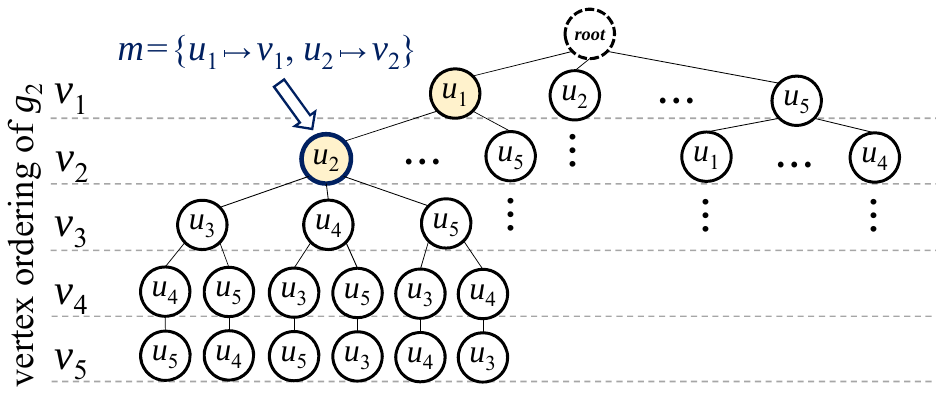}
  \caption{Search tree for graphs in Figure~\ref{fig:mapping}}
\label{fig:prefixtree}
\end{figure}
\vspace*{-0.5cm}

An example search tree for the graphs in Figure~\ref{fig:mapping} is depicted
in Figure~\ref{fig:prefixtree}. In this example, the pre-defined vertex
ordering of $g_2$ is $(v_1, v_2, \ldots, v_5)$.
%
%
Each intermediate node $n$ represents a {\it partial mapping}, which is a
shared prefix of the vertex mappings in the leaves of the subtree rooted by
$n$. Let the $i^{th}$ vertex of $g_2$ be $v$. A tree node containing a vertex
$u$ of $g_1$ at level $i$ represents a mapping $m_p \cup \{u \smapsto v\}$,
where $m_p$ is the mapping of the parent node, and the mapping of the root is
$\emptyset$. In Figure~\ref{fig:prefixtree}, for example, the node indicated by
an arrow corresponds to a partial mapping $m = \{u_1 \smapsto v_1, u_2 \smapsto
v_2\}$.
%
Since a partial mapping uniquely identifies a node in the search tree, we use a
partial mapping interchangeably with the corresponding tree node if clear from
the context.

\begin{defn}[{\bf Lower bound of a partial mapping}]
\label{def:lb}  
The {\it lower bound} of a partial mapping $m$, denoted by ${\sf lb_M}(m)$, is
a lower limit of the edit costs of the vertex mappings in the leaves of the
subtree rooted by $m$.
\end{defn}

Given a GED threshold $\tau$, the subtree rooted by $m$ is pruned if ${\sf lb_M}(m)
> \tau$. To compute ${\sf lb_M}(m)$, we first divide each graph $g$
participating in $m$ into the following three parts:
%
%
\begin{itemize}[leftmargin=*]
\item The {\it mapped subgraph} of $g$ , which is denoted by $g|_m$, is an
  induced subgraph of $g$ defined by the vertices of $g$ participated in
  $m$.
\item The {\it unmapped subgraph} of $g$, which is denoted by $g \backslash
  g|_m$, is an induced subgraph of $g$ defined by the vertices in $V(g)
  \backslash V(g|_m)$.
\item The {\it bridges} are edges connecting $g|_m$ to $g \backslash g|_m$.
\end{itemize}

\noindent
Then, ${\sf lb_M}(m)$ is computed as the sum of
\begin{enumerate}[leftmargin=*]
\item the edit cost required between $g_1|_m$ and $g_2|_m$, which is computed
  as ${\sf ec}(m)$ (Definition~\ref{def:edcost});
\item a lower bound of the GED between $g_1 \backslash g_1|_m$ and $g_2
  \backslash g_2|_m$, which is computed using the {\it label set-based lower
    bound} (Definition~\ref{def:lf});
\item and a lower bound of the number of edit operations required to make the
  bridges of $g_1$ and $g_2$ identical, which is computed using the {\it bridge
    cost} (Definition~\ref{def:bc}).
\end{enumerate}

\begin{defn}[{\bf Label set-based lower bound~\cite{GSIM1}}]
\label{def:lf}  
  The label set-based lower bound between two graphs $r$ and $s$ is defined
  as:
  \[{\sf lb_L}(r, s) = \Gamma(L_V(r), L_V(s)) + \Gamma(L_E(r),
  L_E(s)),\] where $L_V(g)$ and $L_E(g)$ denotes the label multisets of
  vertices and edges of a graph $g$, respectively, and $\Gamma(A, B)$ {\rm =}
  {\sf max}$(|A|, |B|) - |A \cap B|$.
\end{defn}

\begin{defn}[{\bf Bridge cost~\cite{INVES}}]
\label{def:bc}  
  Given a partial mapping $m$, the number of edit operations required in the
  bridges are at least
  \[\mathcal{B}(m) = \sum_{u \rightarrow v \in m} \Gamma(L^{m}_{br}(u), L^{m}_{br}(v)),\]
  where $L^{m}_{br}(w)$ denotes the label multiset of the bridges connected to
  a vertex $w$.
\end{defn}


\begin{exam}
\label{ex:mapping}  
  Consider a partial mapping $m = \{u_1 \smapsto v_1, u_2 \smapsto v_2\}$
  between the two graphs in Figure~\ref{fig:mapping}. The mapped subgraphs,
  unmapped subgraphs, and bridges of the graphs are depicted in blue, black,
  and red lines in the figure, respectively. The edit cost between $g_1|_m$ and
  $g_2|_m$ is 0 since ${\sf ec}(m) = 0$. The label multiset of the vertices in
  $g_1 \backslash g_1|_m$ is $\{B, C, D\}$, which is the same as that in $g_2
  \backslash g_2|_m$. Similarly, the label multiset of the edges in $g_1
  \backslash g_1|_m$ is identical to that in $g_2 \backslash g_2|_m$. Hence,
  ${\sf lb_L}(g_1 \backslash g_1|_m, g_2 \backslash g_2|_m) =
  0$. $\mathcal{B}(m) = 2$ because the bridge label difference between $u_1$
  and $v_1$ is 1 and that between $u_2$ and $v_2$ is also 1. Therefore, ${\sf
    lb_M}(m) = 2$.
\end{exam}

Most of existing solutions traverse the search tree in a best-first fashion
based on the lower bound of each node. Initially, the search tree has a root
node only. They first expand child nodes of the root node of the search
tree. Then, they repeatedly expand child nodes of a node having the lowest
lower bound. It is guaranteed that if they meet a leaf node, the vertex mapping
for the leaf has the minimum edit cost. If all the subtrees of expanded nodes
are pruned, the pair of graphs does not meet the given GED threshold.

\subsection{Related Work}
\label{sec:related}

Existing filtering techniques utilize features of graphs to establish a
necessary condition to meet a GED threshold. Motivated by $q$-gram idea in
string similarity search (e.g., \cite{EDJOIN}), $k$-{\sf AT}~\cite{KAT} defines
a $q$-gram as a tree rooted by a vertex $v$ with all vertices reachable to $v$
in $q$ hops, and {\sf GSimSearch}~\cite{GSIM1, GSIM2} defines a path-based
$q$-gram which is a simple path with length $q$. These techniques are based on
the observation that if the GED between two graphs is within a threshold, then
the graphs should share at least a certain number of $q$-grams. {\sf
  c-star}~\cite{STAR} and {\sf branch}~\cite{BRANCH} structures have been
proposed to derive GED lower bounds through bipartite matching. {\sf c-star} is
1-gram defined by $k$-{\sf AT} and {\sf branch} is a vertex with edges adjacent
to the vertex. All of these filtering techniques have focused on developing
offline index structures. \textsf{SEGOS}~\cite{SEGOS} is a two-level index
structure proposed to efficiently search star structures.

Recent techniques~\cite{PARS, PARS2, MLINDEX, INVES} make use of disjoint
substructures of graphs to capture structural differences between graphs. Based
on the observation in string similarity search (e.g., \cite{PASSJOIN}) and DNA
read mapping techniques (e.g., \cite{HOBBES3}), they decompose each data graph
into partitions, and filter out data graphs dissimilar to the query using the
pigeonhole principle.  {\sf Pars}~\cite{PARS, PARS2} and {\sf
  MLIndex}~\cite{MLINDEX} build offline inverted indices on partitioned
subgraphs. {\sf Mixed}~\cite{MIXED} utilizes small and large disjoint
substructures along with branch structures. {\sf Inves}~\cite{INVES} develops
an online partitioning algorithm that can be used without an index.

The most widely used algorithm for GED computation is {\sf A*-GED}~\cite{GED}.
Recently, {\sf BLP-GED}~\cite{BLPGED}, {\sf DF-GED}~\cite{DFGED}, and {\sf
  CSI\_GED}~\cite{CSIGED,CSIGED2} have been proposed to improve the performance
of GED computation. {\sf BLP-GED} formulates the problem as a binary linear
program, and it is faster and more memory-efficient than {\sf A*-GED}. {\sf
  DF-GED} traverses the search space in a depth first fashion. It has been
found to much more memory-efficient than {\sf A*-GED}. In contrast, {\sf
  CSI\_GED} proposed an edge-based depth-first search. It also has been found
to be both much faster and more memory-efficient than {\sf A*-GED}.

The verification phase of graph similarity search techniques has been developed
based on the {\sf A*-GED} algorithm. {\sf GSimSearch}~\cite{GSIM1} has
suggested that the lower bound computation of {\sf A*-GED} be improved by
utilizing the label set differences. This approach is much faster than the
bipartite heuristic used in {\sf A*-GED}. {\sf Inves}~\cite{INVES} has
introduced a bridge-based lower bound estimation technique, which substantially
reduces the search space.
%
%
{\sf Inves} and {\sf GSimSearch} also exploited effective vertex
orderings for improving the performance of GED computation.

\section{Similarity Search Framework}
\label{sec:framework}

In this section, we present a new approach to graph similarity search. We first
show that candidates of a query can be generated through GED verification, then
develop a novel graph similarity search algorithm. We analyze the proposed
approach to show it can substantially reduce the number of candidates.

\subsection{Candidate Generation Method}
\label{sec:candgen}

In this subsection, we show that candidates of a query can be generated from
the similarity search results of a data graph. After we define a candidate set
in Definition~\ref{def:candset}, we formally state our observation in
Lemma~\ref{lm:candgen}.

\begin{defn}[{\bf Candidate set}]
\label{def:candset}  
  For a graph database $\mathcal{D}$ and a query graph $q$ with a GED threshold
  $\tau$, a candidate set $\mathcal{C}$ of the query is any subset of
  $\mathcal{D}$ that satisfies $\mathcal{R}(q, \tau) \subseteq \mathcal{C}$.
\end{defn}

\begin{lmma}
\label{lm:candgen}  
  Consider a graph database $\mathcal{D}$ and a query graph $q$ with a GED
  threshold $\tau$. Given a data graph $g \in \mathcal{D}$ such that
  $\mathsf{ged}(q, g) = \delta \leq \tau$, the following inclusion relationships hold.
  
  \begin{enumerate}[leftmargin = 3cm]
    \item $\mathcal{R}(q, \tau) \subseteq \mathcal{R}(g, \tau + \delta)$
    \item $\mathcal{R}(g, \tau - \delta) \subseteq \mathcal{R}(q, \tau)$
  \end{enumerate}  
\end{lmma}
\begin{proof}
  For every graph $r\in \mathcal{R}(q, \tau)$, the triangle inequality
  $\mathsf{ged}(g, r) \leq \mathsf{ged}(q, r) + \mathsf{ged}(q, g)$ holds by
  Lemma~\ref{lm:metric}.  Since $\mathsf{ged}(q, r) \leq \tau$ and
  $\mathsf{ged}(q, g) = \delta$, $\mathsf{ged}(g, r) \leq \tau + \delta$, which
  implies $r \in \mathcal{R}(g, \tau + \delta)$. It can be similarly proved
  that every graph $r' \in \mathcal{R}(g, \tau - \delta)$ is included in
  $\mathcal{R}(q, \tau)$.
\end{proof}

Lemma~\ref{lm:candgen} suggests that as soon as we identify a result of the
query, we can generate a candidate set. It also suggests that some candidate
graphs can be directly determined as results of a query without any
verification. For every pair of graphs $g_1$ and $g_2$ in $\mathcal{D}$, {\sf
  ged}$(g_1, g_2)$ can be pre-computed and materialized to immediately obtain
$\mathcal{R}(g, \tau + \delta)$ and $\mathcal{R}(g, \tau - \delta)$ stated in
Lemma~\ref{lm:candgen}. We use the pre-computed results as an index for
generating candidates of a query.  For example, Figure~\ref{fig:geds} shows
pre-computed GEDs for the graphs in Figure~\ref{fig:graph}. We will discuss how
to implement such an index in Section~\ref{sec:index}.
Example~\ref{ex:candgen} demonstrates our candidate generation.

\begin{figure}[htbp]
  \centering \includegraphics[height=4cm]{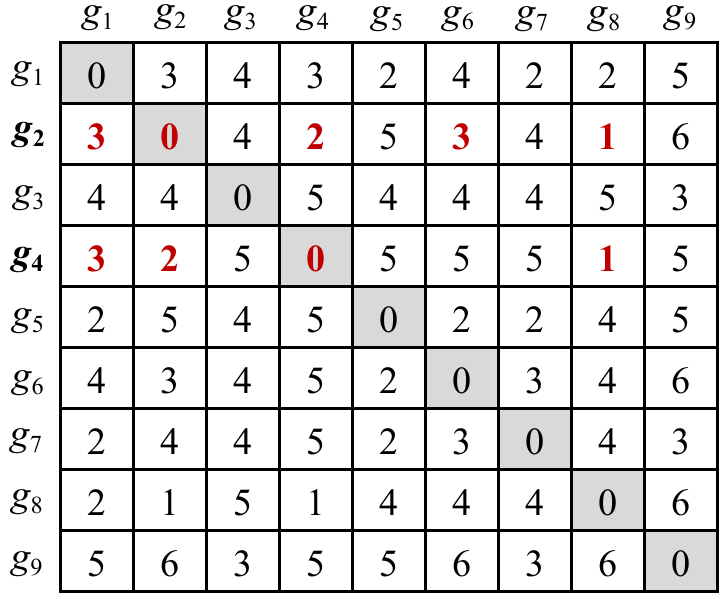}
  \caption{GEDs between graphs in Figure~\ref{fig:graph}}
\label{fig:geds}
\end{figure}

\begin{exam}
\label{ex:candgen}
  In Figure~\ref{fig:graph}, consider the query graph $q$ with a GED threshold
  $\tau = 2$.  $g_2$ is a result of the query because {\sf ged}$(q, g_2) =
  \delta = 1$. As soon as $g_2$ is identified as a result, $\mathcal{R}(g_2,
  \tau + \delta = 3) = \{g_1, g_2, g_4, g_6, g_8\}$ becomes a candidate set of
  the query by Lemma~\ref{lm:candgen}. Among the candidates, $g_8$ is
  identified as a result of the query without verification, because
  $\mathcal{R}(g_2, \tau - \delta = 1) = \{g_2, g_8\}$ (by
  Lemma~\ref{lm:candgen}).
\end{exam}

\subsection{Search Algorithm}
\label{sec:nass}

For a query graph $q$ with a GED threshold $\tau$, let $\mathcal{A}$ be a set
of results identified so far (initially, $\mathcal{A} = \emptyset$). For the
ease of presentation, we abuse a candidate set to denote any subset
$\mathcal{C}$ of the database $\mathcal{D}$ that satisfies $(\mathcal{R}(q,
\tau) - \mathcal{A}) \subseteq \mathcal{C}$.

Using the label set-based lower bound in Definition~\ref{def:lf}, we first
generate an initial candidate set $\mathcal{C}_0 = \{g\ |\ g \in \mathcal{D}
\land {\sf lb_L}(q, g)~\leq~\tau\}$. Then, we repeatedly regenerate a candidate
set using Lemma~\ref{lm:candgen} whenever we find a result from the current
candidate set, which is initially $\mathcal{C}_0$. Definition~\ref{def:regen}
and Lemma~\ref{lm:regen} formally state the candidate regeneration.
%
%
In the following description, we assume that if we find a result $r$ such that
${\sf ged}(q, r) = \delta$, we also immediately identify the results in
$\mathcal{R}(r, \tau - \delta)$ by Lemma~\ref{lm:candgen}.

\begin{defn}[{\bf Candidate regeneration}]
\label{def:regen}  
  Given a candidate set $\mathcal{C}$, let $r$ be the first result identified
  in $\mathcal{C}$, and ${\sf ged}(q, r)$ be $\delta$. The refined candidate
  set ${\sf RC}(\mathcal{C})$ is defined as:
\[{\sf RC}(\mathcal{C}) = (\mathcal{C} - \mathcal{V}(\mathcal{C})) \cap
(\mathcal{R}(r, \tau+\delta) - \mathcal{R}(r, \tau-\delta)),\]
where $\mathcal{V}(\mathcal{C})$ denotes the set of those candidates in
$\mathcal{C}$ that are verified until the first result $r$ is identified.
\end{defn}

\begin{lmma}
  \label{lm:regen}
  The refined candidate set ${\sf RC}(\mathcal{C})$ in
  Definition~\ref{def:regen} contains all remaining results.
\end{lmma}
\begin{proof}
  It is obvious that $\mathcal{C} - \mathcal{V}(\mathcal{C})$ contains all
  remaining results. Because, by Lemma~\ref{lm:candgen}, $\mathcal{R}(r,
  \tau+\delta)$ is a candidate set of the query and $\mathcal{R}(r,
  \tau-\delta)$ contains identified results, $\mathcal{R}(r, \tau+\delta) -
  \mathcal{R}(r, \tau-\delta)$ also contains all remaining results. Therefore,
  $(\mathcal{C} - \mathcal{V}(\mathcal{C})) \cap (\mathcal{R}(r, \tau+\delta) -
  \mathcal{R}(r, \tau-\delta))$ contains all remaining results.
%
\end{proof}

In our method, the total number of candidates is dynamically determined because
candidates are repeatedly regenerated while processing a
query. Lemma~\ref{lm:candset} gives a formula to compute the set of total
candidates that are verified through GED computation, and
Corollary~\ref{coro:res} states how to obtain the result set of the query.
%

\begin{lmma}
\label{lm:candset}
Given a candidate set $\mathcal{C}$, 
the set of all candidates that are verified through GED computation is
  %
  \[{\sf NassCand}(\mathcal{C}) = \mathcal{V}(\mathcal{C}) \cup
  {\sf NassCand}({\sf RC}(\mathcal{C})),\]
  where {\sf NassCand}$(\emptyset) = \emptyset$.
\end{lmma}
\begin{proof}
  $\mathcal{V}(\mathcal{C})$ contains all candidates already verified, and
  $\sf{RC}(\mathcal{C})$ contains all results requiring GED verification by
  Lemma~\ref{lm:regen}.  Therefore, it can be proved by induction that ${\sf
    NassCand}(\mathcal{C})$ contains all candidates requiring GED verification.
\end{proof}

\newcommand*{\medcup}{\mathbin{\scalebox{1.5}{\ensuremath{\cup}}}}%
\begin{coro}
\label{coro:res}  
  Let $\{(r_1, \delta_1), \ldots, (r_n, \delta_n)\}$ be the results identified
  while verifying candidates in ${\sf NassCand}(\mathcal{C}_0)$, where
  $\mathcal{C}_0$ is an initial candidate set. The result set of the query is
  $\medcup_{i=1}^n \mathcal{R}(r_i, \tau - \delta_i)$.
\end{coro}
\begin{proof}
  $\mathcal{R}(q, \tau) = \medcup_{i=1}^{n} (\{r_i\} \cup \mathcal{R}(r_i, \tau
  - \delta_i))$ by Lemma~\ref{lm:regen} and Lemma~\ref{lm:candset}.  $\forall
  i\ r_i \in \mathcal{R}(r_i, \tau - \delta_i)$, since ${\sf ged}(r_i, r_i) =
  0$ by Lemma~\ref{lm:metric} and $0 \leq \delta_i \leq \tau$. Therefore,
  $\medcup_{i=1}^{n} \{r_i\} \subseteq \medcup_{i=1}^n \mathcal{R}(r_i, \tau -
  \delta_i)$.
\end{proof}  


\begin{exam}
  \label{ex:candset}
  Consider the query in Example~\ref{ex:simsearch}, again. The initial
  candidate set is $\mathcal{C}_0 = \{g_1, \ldots, g_8\}$, since $\forall i \in
  [1,8]\ {\sf lb_L}(q, g_i) \leq \tau$ but ${\sf lb_L}(q, g_9) = 4 > \tau$,
  where $\tau = 2$.  To evaluate the query, we verify $g_1$ and $g_2$ to find
  the first result $g_2$ because ${\sf ged}(q, g_2) = 1$. Therefore,
  $\mathcal{V}(\mathcal{C}_0) =\{g_1, g_2\}$. After $g_2$ is identified as a
  result, by Definition~\ref{def:regen}, remaining candidates of the query are
  refined to
  \begin{align*}
    {\sf RC}(\mathcal{C}_0) &= (\mathcal{C}_0 - \mathcal{V}(\mathcal{C}_0))
    \cap (\mathcal{R}(g_2, 2+1) - \mathcal{R}(g_2, 2-1))\\
    &= (\{g_1, \ldots, g_8\} - \{g_1, g_2\}) \cap (\{g_1, g_2, g_4, g_6, g_8\} - \{g_2, g_8\})\\
    &= \{g_4, g_6\}.
  \end{align*}
  Let $\mathcal{C}_1 = RC(\mathcal{C}_0)$. Given the refined candidate set
  $\mathcal{C}_1$, we identify $g_4$ as the first result in $\mathcal{C}_1$
  because ${\sf ged}(q, g_4) = 2$. Hence, $\mathcal{V}(\mathcal{C}_1) =
  \{g_4\}$ and $\mathcal{C}_1$ is refined to
  \begin{align*}
    {\sf RC}(\mathcal{C}_1) &= (\mathcal{C}_1 - \mathcal{V}(\mathcal{C}_1))
    \cap (\mathcal{R}(g_4, 2+2) - \mathcal{R}(g_4, 2-2))\\
    &= (\{g_4, g_6\} - \{g_4\}) \cap (\{g_1, g_2, g_4, g_8\} - \{g_4\}) 
    = \emptyset.
  \end{align*}
  By Lemma~\ref{lm:candset}, all the verified candidates are contained in
  \begin{align*}
    {\sf NassCand}(\mathcal{C}_0) &= \mathcal{V}(\mathcal{C}_0) \cup
    {\sf NassCand}({\sf RC}(\mathcal{C}_0))\\
    &= \mathcal{V}(\mathcal{C}_0) \cup \mathcal{V}(\mathcal{C}_1) \cup
    {\sf NassCand}({\sf RC}(\mathcal{C}_1))\\
    &= \{g_1, g_2\} \cup \{g_4\} \cup {\sf NassCand}(\emptyset)\\
    &= \{g_1, g_2, g_4\}.
  \end{align*}
  Since $(g_2, 1)$ and $(g_4, 2)$ are identified results from ${\sf
    NassCand}(\mathcal{C}_0)$, the result set of the query is $\mathcal{R}(g_2,
  2 - 1) \cup \mathcal{R}(g_4, 2 - 2) = \{g_2, g_8\} \cup \{g_4\} = \{g_2, g_4,
  g_8\}$ by Corollary~\ref{coro:res}.
%
  %
\end{exam}

\begin{algorithm}[htbp]
\SetKwInOut{input}{input}
\SetKwInOut{output}{output}
\input{$\mathcal{C}$ is a candidate set,\\
  $q$ is a query graph, and $\tau$ is a GED threshold}
\output{query results in $\mathcal{C}$}
\BlankLine
 sort graphs in $\mathcal{C}$ by their GED lower bounds\;
 \BlankLine

 $\mathcal{V} \leftarrow \emptyset$\;
 \ForEach{\rm{candidate} $g \in \mathcal{C}$}{
   $\delta \leftarrow$ {\sf NassGED}$(q, g, \tau)$; \tcp*[f]{refer to Algorithm~\ref{alg:ged}}\\
   $\mathcal{V}$ $\leftarrow$ $\mathcal{V} \cup \{g\}$\;
   \If{$\delta \leq \tau$}{
     $\mathcal{A} \leftarrow \mathcal{R}(g, \tau - \delta)$\;
     
     $\mathcal{C}' \leftarrow (\mathcal{C} - \mathcal{V}) \cap
     (\mathcal{R}(g, \tau + \delta) - \mathcal{A})$\;
      \Return $\mathcal{A}\ \cup$ {\sf Nass}($\mathcal{C}', q, \tau$)\;}}
 \Return $\emptyset$\;
\caption{\textsf{Nass}($\mathcal{C}$, $q$, $\tau$)}
\label{alg:nass}
\end{algorithm}

Algorithm~\ref{alg:nass} outlines our graph similarity search algorithm based
on the proposed candidate generation method. Initially, the algorithm is called
with a candidate set generated using the label set-based lower bound.
%
%
Given a candidate set $\mathcal{C}$, the algorithm sorts candidates by their
GED lower bounds (Line~1). By sorting the candidates, it first verifies those
candidates that are more likely to be results. It computes the GED between the
query and each candidate using our GED computation algorithm {\sf NassGED},
which will be presented in Section~\ref{sec:GEDalgorithm} (Line~4).
%
%
Candidates verified until the first result is found are appended into
$\mathcal{V}$ (Line~5).
%
If the algorithm encounters the first result (Line~6), it appends the results
in $\mathcal{R}(g, \tau - \delta)$ into $\mathcal{A}$ by
Corollary~\ref{coro:res} (Line~7), and refines remaining candidates based on
Definition~\ref{def:regen} (Line~8). Then, it continues to verify and refine
candidates (Line~9). If it cannot find any result from $\mathcal{C}$, it
returns an empty set (Line~10).



\vspace*{0.7em}
\noindent
{\bf Correctness of Algorithm~\ref{alg:nass}.}  Since the algorithm scans the
candidate set $\mathcal{C}$ sequentially, it always collects the first result
along with those results that do not require GED verification by
Lemma~\ref{lm:candgen} (Line~7). Then, it regenerates a candidate set which is
assured to contain all remaining results requiring verification by
Lemma~\ref{lm:regen} (Line~8). Therefore, it can be proved by induction that
the algorithm correctly collects all results.

\subsection{Analysis of Our Algorithm}
\label{sec:analysis}

We analyze our algorithm by estimating the number of candidates requiring GED
verification. Before we estimate it, we briefly review candidate sets generated
by existing solutions. Given a query $q$ with a GED threshold $\tau$, all
existing candidate generation techniques use a GED lower bound function ${\sf
  f_{lb}}$ to generate a candidate set $\mathcal{C}_{\sf f_{lb}}(q, \tau) =
\{g\ |\ g \in \mathcal{D} \land {\sf f_{lb}}(q, g) \leq \tau
\}$. Proposition~\ref{po:inclusion} states the relationship between candidate
sets generated by a GED lower bound function ${\sf f_{lb}}$.
\begin{prop}
  \label{po:inclusion}
  For any GED lower bound function ${\sf f_{lb}}$, the following implication
  holds:
  \[\forall \tau_1, \tau_2 \ \ \tau_1 \leq \tau_2 \implies \mathcal{C}_{\sf f_{lb}}(q, \tau_1)
  \subseteq \mathcal{C}_{\sf f_{lb}}(q, \tau_2).\]
\end{prop}
\begin{proof}
For every candidate $g \in \mathcal{C}_{\sf f_{lb}}(q, \tau_1)$, ${\sf
  f_{lb}}(q, g) \leq \tau_1 \leq \tau_2$. Therefore, $g \in \mathcal{C}(q,
\tau_2)$.
\end{proof}  
Given an initial candidate set, we assume that all results of the query are
uniformly distributed in the candidate set.  Based on the assumption,
Lemma~\ref{lm:scan} estimates the expected number of candidates requiring GED
computation.

\begin{lmma}
\label{lm:scan}    
Given an initial candidate set $\mathcal{C}_{\sf lb_L}(q, \tau)$ generated
using the label set-based lower bound ${\sf lb_L}$, the expected number of
candidates generated by {\sf Nass} that require GED computation is as follows:
  \[|{\sf NassCand}(\mathcal{C}_{\sf lb_L}(q,\tau))| < \frac{|\mathcal{C}_{\sf lb_L}(q, \delta_{min})|}
     {|\mathcal{R}(q, \delta_{min})| + 1} + |\mathcal{R}(r, \tau +
     \delta_{min})|,\]
where $\delta_{min} = min_{g \in \mathcal{R}(q, \tau)}\ {\sf ged}(q, g)$ and
$r$ is a result of the query such that ${\sf ged}(q, r) = \delta_{min}$.
\end{lmma}
\begin{proof}
  Since our algorithm sorts initial candidates by their GED lower bounds,
  candidates in $\mathcal{C}_{\sf lb_L}(q, \delta_{min})$ are verified first
  by Proposition~\ref{po:inclusion}.
  The first result whose distance is $\delta$ should be contained in
  $\mathcal{C}_{\sf lb_L}(q, \delta_{min})$, where $\delta_{min} \leq \delta \leq
  \tau$. Because there are at least $|\mathcal{R}(q, \delta_{min})|$ results in
  $\mathcal{C}_{\sf lb_L}(q, \delta_{min})$, by the uniformity assumption, the algorithm
  can find the first result after verifying at most $n_{\mathcal{V}} =
  \frac{|\mathcal{C}_{\sf lb_L}(q, \delta_{min})|}{|\mathcal{R}(q, \delta_{min})| + 1}$
  candidates.
  By the assumption again, a result $r$ whose distance is $\delta_{min}$ should
  be found after verifying $n_{\mathcal{V}}$ candidates. Therefore, the number
  of regenerated candidates is less than $|\mathcal{R}(r, \tau +
  \delta_{min})|$ by Lemma~\ref{lm:candgen} and Lemma~\ref{lm:regen}.
\end{proof}  

Based on Lemma~\ref{lm:scan}, the following example estimates the number of
candidates verified by {\sf Nass} using the empirical results in
Table~\ref{tbl:expr-intro-cand}. We note that the candidate generation method
using the label set-based lower bound ${\sf lb_L}$ is the label filtering
method {\sf LF} in Table~\ref{tbl:expr-intro-cand}.
\begin{exam}
\label{ex:numcand}
Let's consider the case where $\tau = 4$ and $\delta_{min} = \tau - 1 = 3$. In
Table~\ref{tbl:expr-intro-cand}, $|\mathcal{C}_{\sf lb_L}(q, \delta_{min})| =
285$ and $|\mathcal{R}(q, \delta_{min})| = 1.26$.  Therefore, we can expect
285/(1+1.26) = 126.1 candidates verified until the first result identified.
Because the average number of results on threshold $\tau + \delta_{min}$ is
$18.45$ in Table~\ref{tbl:expr-intro-cand}, the expected total number of
candidates verified by our algorithm is at most $126.1 + 18.45 = 144.55$, which
much less than that of existing techniques on that threshold (i.e., $\tau =
4$).
\end{exam}
We remark that any existing filtering method, i.e., lower bound
functions, can be used to generate an initial candidate set for {\sf Nass}. If
we use {\sf Inves} filtering method, for example, the expected number of
candidates becomes 44.55 in Example~\ref{ex:numcand}.  Nonetheless, we use a
basic filtering method {\sf LF} in generating initial candidates because our
GED computation algorithm in Section~\ref{sec:ged} integrates existing
filtering techniques.

\section{GED Computation}
\label{sec:ged}

In this section, we first introduce our GED computation model. Then, we propose
a filtering pipeline for GED computation. We finally present the details of our
GED computation algorithm.

\subsection{Motivation and GED Computation Model}
\label{sec:gedmodel}

As presented in Section~\ref{sec:gedcomputation}, existing solutions compute
a lower bound of a partial mapping $m$ between two graphs $g_1$ and $g_2$ as:
%
\[{\sf lb_M}(m) = {\sf ec}(m) + \mathcal{B}(m) + {\sf lb_L}(g_1 \backslash g_1|_m, g_2
\backslash g_2|_m).\]
In the formula, ${\sf ec}(m)$ is the tight bound and $\mathcal{B}(m)$ is a
relatively precise bound. However, the label set-based lower bound ${\sf
  lb_L}$, which is used for the unmapped subgraphs, is very loose because it
does not take structural differences into considerations. As a consequence,
existing solutions suffer from a huge search space. To address the problem,
throughout Section~\ref{sec:ged}, we focus on tightening the lower bound
between the unmapped subgraphs by introducing a lower bound function, which
exploits a few existing GED lower bounds.

Let ${\sf f_{lb}}$ be a GED lower bound function. If ${\sf lb_M}(m) = {\sf
  ec}(m) + \mathcal{B}(m) + {\sf f_{lb}}(g_1 \backslash g_1|_{m}, g_2
\backslash g_2|_{m}) > \tau$, by Definition~\ref{def:lb}, we can prune the
subtree rooted by $m$ from the search tree. By rewriting the inequality
focusing on ${\sf f_{lb}}$, we establish the following filtering condition.
\begin{cond}
  \label{cond:filter}
Given a partial mapping $m$, we can prune the subtree rooted by $m$ if 
${\sf f_{lb}}(g_1 \backslash g_1|_m, g_2 \backslash g_2|_{m}) >
\tau - ({\sf ec}(m) + \mathcal{B}(m)).$
\end{cond}
Using Condition~\ref{cond:filter}, we model GED computation as a repetition of
filtering dissimilar unmapped subgraphs while traversing the search tree. To
efficiently obtain a tight ${\sf f_{lb}}(g_1 \backslash g_1|_m, g_2 \backslash
g_2|_{m})$, in Section~\ref{sec:filtering}, we judiciously select and carefully
apply a series of existing feature-based lower bound functions, which have been
used in generating candidates.
As pointed out in Section~\ref{sec:intro}, existing feature-based filtering
techniques have a limitation in filtering dissimilar graphs. Nevertheless, we
observe that they can be effectively used in pruning dissimilar unmapped
subgraphs for GED computation as stated in Claim~\ref{claim:filter}.

\begin{claim}
\label{claim:filter}  
  Given a partial mapping $m$ between $g_1$ and $g_2$, 
  \[{\sf Pr}[{\sf f_{lb}}(g_1, g_2) > \tau] \leq {\sf Pr}[{\sf f_{lb}}(g_1 \backslash g_1|_{m}, g_2 \backslash g_2|_{m}) > \tau'] \]
  for any lower bound function ${\sf f_{lb}}$, where ${\sf Pr}[p]$ denotes the
  probability that $p$ is true, and $\tau' =\tau - ({\sf ec}(m) +
  \mathcal{B}(m))$.
\end{claim}
\begin{proof}
  In this proof, we use an approximate assumption\footnote{Due to
    the inaccuracy of the bridge cost, there can be subtle cases that ${\sf
      lb_M}(m_1) > {\sf lb_M}(m_2)$, but the assumption is valid in most cases for
    any lower bound function ${\sf f_{lb}}$. For example, in our all experiments
    in Section~\ref{sec:expr}, there was no mapping that violates the
    assumption (the total number of different mappings was about $2.5 \times
    10^7$ in our experiments).}  that ${\sf lb_M}(m_1) \leq {\sf lb_M}(m_2)$ for
  any mappings $m_1$ and $m_2$ such that $m_1$ is a prefix of $m_2$. Since an empty
  mapping $\emptyset$ is a prefix of any mapping, by the assumption,
  \begin{align*}
    {\sf lb_M}(\emptyset) &= {\sf f_{lb}}(g_1, g_2)\\
    &\leq {\sf ec}(m) + \mathcal{B}(m) + {\sf f_{lb}}(g_1 \backslash g_1|_{m}, g_2 \backslash g_2|_{m}) = {\sf lb_M}(m).
  \end{align*}
  Thus, $\tau - {\sf f_{lb}}(g_1, g_2) \geq  \tau' - {\sf f_{lb}}(g_1 \backslash g_1|_{m}, g_2 \backslash g_2|_{m}))$, which implies ${\sf Pr}[{\sf f_{lb}}(g_1, g_2) > \tau] \leq {\sf Pr}[{\sf f_{lb}}(g_1 \backslash g_1|_{m}, g_2 \backslash g_2|_{m}) > \tau']$.
\end{proof}

Condition~\ref{cond:filter} also enables us to design a new GED computation
algorithm that seamlessly integrates the filtering phase into GED
computation. If we apply Condition~\ref{cond:filter} to the root node of the
search tree (i.e., $m=\emptyset$), the condition becomes ${\sf f_{lb}}(g_1,
g_2) > \tau$, which is the condition used in the filtering phase of existing
search techniques. For example, if we use the online partitioning-based lower bound of
{\sf Inves}~\cite{INVES} as ${\sf f_{lb}}$ and apply it to the root node, we
can make a GED computation algorithm that encompasses the candidate refinement
step of {\sf Inves}. We remark that existing GED algorithms compute the lower
bound of a mapping that has at least one mapped vertex pair.
We will present the details of our GED computation algorithm in
Section~\ref{sec:GEDalgorithm}.

\COMMENTOUT{
Another difference with existing solutions is that we seamlessly integrate the
filtering phase into GED computation by applying Condition~\ref{cond:filter} to
the root node (i.e., $m = \emptyset$) of the search tree. When $m = \emptyset$,
Condition~\ref{cond:filter} becomes ${\sf f_{lb}}(g_1, g_2) > \tau$, which is
the condition used in existing filtering techniques. For example, if we use the
online partitioning-based lower bound of {\sf Inves}~\cite{INVES} as ${\sf
  f_{lb}}$, and apply it to the root node, we can make a GED computation
algorithm that encompasses the candidate refinement step of {\sf Inves}. We
remark that existing GED algorithms compute the lower bound of a mapping that
has at least one mapped vertex pair.
We will present the details of our GED computation algorithm in
Section~\ref{sec:GEDalgorithm}.
}

\subsection{Filtering Pipeline in GED Computation}
\label{sec:filtering}

To tighten ${\sf f_{lb}}(g_1 \backslash g_1|_{m}, g_2 \backslash g_2|_{m})$, we
apply a series of filtering techniques. There is a trade-off between the
efficiency in computing ${\sf f_{lb}}$ and the tightness of ${\sf f_{lb}}$.
Because the number of nodes to visit in the search tree grows exponentially,
efficient computation of a lower bound is crucial. With tight lower bounds, on
the other hand, we can prune more subtrees in the search tree, and the number
of nodes to visit can be reduced substantially. Therefore, the goal here is to
judiciously select filtering techniques adequate for reducing the search space,
and to carefully apply selected filters for efficient computation. To speed up
the computation of ${\sf f_{lb}}$, we will also discuss implementation issues
in Section~\ref{sec:gedimpl}.

We first introduce two existing lower bound functions we select for GED
computation. Given a partial mapping $m$, for simplicity, we use $g'$ to denote
$g \backslash g|_{m}$ for a graph $g$ in this subsection.

\begin{defn}[{\bf Compact branch-based lower bound~\cite{MIXED}}]
\label{def:cbranch}  
  Giv\-en two vertices $u$ and $v$, their branch structures are denoted as $b_u =
  (l(u), ES(u))$ and $b_v = (l(v), ES(v))$, where $ES(w) =
  \{l(e)\ |\ \text{edge } e \text{ is}$ $\text{adjacent to } w\}$. The compact
  distance between $b_u$ and $b_v$ is defined as:
  \[{\sf bed_C}(b_u, b_v) =
  \begin{cases}
    0,   & \text{\rm if } l(u) = l(v) \land ES(u) = ES(v) \\
    1/2, & \text{\rm if } l(u) = l(v) \land ES(u) \neq ES(v)\\
    1,   & \text{\rm if } l(u) \neq l(v).
  \end{cases}\]
  Compact branch-based lower bound is defined as:
  \[{\sf lb_C}(g'_1, g'_2) = \min_P \sum_{b_u \in B(g'_1)} {\sf bed_C}(b_u, P(b_u)),\]
  where $B(g)$ is the multiset of the branches of a graph $g$, and $P$ is a
  bijection from $B(g'_1)$ to $B(g'_2)$. If $|B(g'_1)| < |B(g'_2)|$, $|B(g'_2)| - |B(g'_1)|$
  blank branches are added into $B(g'_1)$, and vice versa.
\end{defn}
\begin{figure}[htbp]
  \centering
  \subfigure[]{\includegraphics[height=1.7cm]{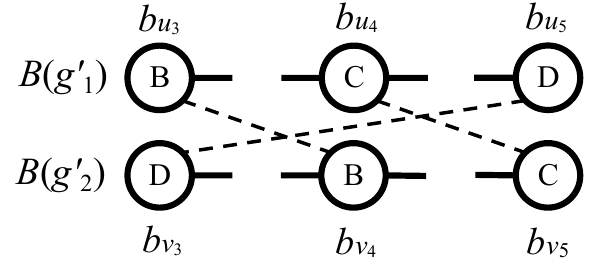}}\hspace{0.6cm}
  \subfigure[]{\includegraphics[height=1.7cm]{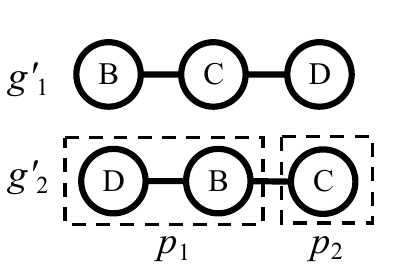}}
  \caption{Branches and partitioning of $g'_1$ and $g'_2$}
  \label{fig:branchpartition}
\end{figure}
\begin{exam}
\label{ex:branch}  
  Let's recall Example~\ref{ex:mapping} in Section~\ref{sec:gedcomputation}.
  The label set-based lower bound of the unmapped subgraphs was 0. This lower
  bound can be tightened by using the compact branch-based lower bound function
  as follows. Figure~\ref{fig:branchpartition}(a) shows the branch multisets
  $B(g'_1) = \{b_{u_3}, b_{u_4}, b_{u_5}\}$ and $B(g'_2) = \{b_{v_3}, b_{v_4},
  b_{v_5}\}$, and the bijection $P = \{b_{u_3} \smapsto b_{v_4}, b_{u_4}
  \smapsto b_{v_5}, b_{u_5} \smapsto b_{v_3}\}$ that minimizes ${\sf
    lb_C}(g'_1, g'_2)$. The lower bound ${\sf lb_C}(g'_1, g'_2) = 1$ because
  ${\sf bed_C}(b_{u_3}, b_{v_4}) = 1/2$, ${\sf bed_C}(b_{u_4}, b_{v_5}) = 1/2$,
  and ${\sf bed_C}(b_{u_5}, b_{v_3}) = 0$.
\end{exam}  
To compute ${\sf lb_C}(g'_1, g'_2)$, we can use an $O(n\log n)$ algorithm
proposed in \cite{MIXED}, where $n = |B(g'_1)|$. The compact branch-based lower
bound is used for generating candidate graphs in \cite{MIXED}.
Since it is efficiently computed and captures differences in local structures
(i.e., branches) of graphs, we select it for GED computation.

\begin{defn}[{\bf Partition-based lower bound~\cite{INVES}}]
\label{def:partition}    
    Consider we decompose $g'_2$ into partitioned subgraphs. The
    partition-based lower bound is defined as:
    \[{\sf lb_P}(g'_1, g'_2) = |\{p\ |\ p \in \mathcal{P}(g'_2) \land p
    \not\sqsubseteq g'_1\}|,\]
    where $\mathcal{P}(g'_2)$ denotes the set of partitions of $g'_2$, and
    $p \not\sqsubseteq g'_1$ denotes $p$ is not subgraph isomorphic to $g'_1$.
\end{defn}

\begin{exam}
  In Example~\ref{ex:mapping}, consider we decompose $g'_2$ into two
  partitions $p_1$ and $p_2$ as depicted in
  Figure~\ref{fig:branchpartition}(b). Since $p_1 \not\sqsubseteq g'_1$ and
  $p_2 \sqsubseteq g'_1$, the lower bound ${\sf lb_P}(g'_1, g'_2) = 1$.
\end{exam}  

To partition $g'_2$ in Definition~\ref{def:partition}, we use the online
partitioning algorithm proposed in {\sf Inves}~\cite{INVES}\footnote{To save
  computation time, we modified {\sf Inves} by disabling the {\it rematch}
  functionality and setting $\alpha$, i.e., the worst case prevention
  parameter, to 6 (see \cite{INVES} for the details).}. With partitions of
graphs, we can capture structural differences between graphs, and thus we can
expect a more accurate bound in general.
However, it is expensive to compute ${\sf lb_P}$ due to subgraph isomorphism
tests. In this paper, therefore, we use ${\sf lb_P}$ only when other lower
bound functions cannot filter out $g'_1$ and $g'_2$.

Given a partial mapping $m$ and a GED threshold $\tau$, to reduce the
overhead of computing lower bounds, we incrementally tighten the lower bound of
$m$ as follows.
\[{\sf lb_M}(m)=
  \begin{cases}
    {\sf ec}(m), & \text{\phantom{****} if } {\sf ec}(m) > \tau\\
    {\sf ec}(m) + \mathcal{B}(m), & \text{else if } {\sf ec}(m) + \mathcal{B}(m) > \tau\\
    {\sf ec}(m) + \mathcal{B}(m)&\\
     \phantom{--..} + {\sf f_{lb}}(g'_1, g'_2), & \text{otherwise,}
  \end{cases}
\]
and ${\sf f_{lb}}(g'_1, g'_2)$ is defined as:
\[{\sf f_{lb}}(g'_1, g'_2)=
  \begin{cases}
    {\sf lb_L}(g'_1, g'_2), & \text{\phantom{****} if } {\sf lb_L}(g'_1, g'_2) > \tau'\phantom{---}\\
    {\sf lb_C}(g'_1, g'_2), & \text{else if } {\sf lb_C}(g'_1, g'_2) > \tau'\\
    {\sf lb_P}(g'_1, g'_2), & \text{otherwise,}
  \end{cases}
\]
where $\tau' =  \tau - ({\sf ec}(m) + \mathcal{B}(m))$.\\[-0.6em]

Interestingly, many partial mappings in the search tree have the same unmapped
subgraphs. For the graphs in Figure~\ref{fig:mapping}, for example, $m = \{u_1
\smapsto v_1, u_2 \smapsto v_2\}$ and $m' = \{u_1 \smapsto v_2, u_2 \smapsto
v_1\}$ have the same unmapped subgraphs. Therefore, we can compute ${\sf
  f_{lb}}(g'_1, g'_2)$ once and share the result in those partial mappings having
the same unmapped subgraphs.
Lemma~\ref{lm:unmappedpart} states the
number of partial mappings having the same unmapped subgraphs.

\begin{lmma}
\label{lm:unmappedpart}
Given a partial mapping $m$ between two graphs $g_1$ and $g_2$, there are $|m|!
/ n_{\varepsilon}!$ partial mappings in the search tree that have the same
unmapped subgraphs, where $n_{\varepsilon}$ is the number of copies of
$\varepsilon$ in $V(g_1|_{m})$.
\end{lmma}
\begin{proof}
  Given two graphs for GED computation, any graph can be $g_2$ by the symmetry
  in Lemma~\ref{lm:metric}. Therefore, we assume, without loss of generality,
  there is no $\varepsilon$ in $V(g_2)$.  Consider a partial mapping $m'$
  between $g_1$ and $g_2$ such that $V(g_1|_{m}) = V(g_1|_{m'})$. By the
  definition of the unmapped subgraph (see Section~\ref{sec:gedcomputation}), $m$
  and $m'$ have the same unmapped subgraph of $g_1$.
  Recall that we use a specific ordering of $g_2$ for all mappings.
  Therefore, $m'$ and $m$ also have the same vertex set for $g_2$, and the same
  unmapped subgraph of $g_2$. There are $|m|! / n_{\varepsilon}!$ distinct
  permutations of $V(g_1|_{m})$, and thus there are $|m|! / n_{\varepsilon}!$
  mappings in the search tree that have the same unmapped subgraphs.
\end{proof}

\subsection{GED Computation Algorithm}
\label{sec:GEDalgorithm}

Given a partial mapping $m$, we compute ${\sf lb_M}(m)$ as shown in
Algorithm~\ref{alg:lowerbound}. If $m$ survives from ${\sf ec}(m)$ and
$\mathcal{B}(m)$ in our filtering pipeline (Lines~1--4), we look up a hash with
$V(g_1|_{m})$ to share the computation result of ${\sf f_{lb}}(g'_1, g'_2)$, if
any, based on Lemma~\ref{lm:unmappedpart} (Line~5). We retrieve a pointer to
the hash entry, where the hash entry $e$ has $e.lb = {\sf f_{lb}}(g'_1, g'_2)$
and an index $e.index$ for the lower bound function used in computing ${\sf
  f_{lb}}(g'_1, g'_2)$ (i.e., one of ${\sf lb_L}$, ${\sf lb_C}$, and ${\sf
  lb_P}$).
If the lookup fails, the hash makes a new entry $e$ such that $e.lb = 0$ and
$e.index = 0$, and return the pointer to the entry.
If $m$ also survives from the lower bound from the hash (Lines~5--6) and not
all lower bound functions are applied (Line~9), we apply unused lower bound
functions to $m$ (Lines~9--14), and update the hash entry if necessary (Line~10
and Line~12).

\begin{algorithm}[t]
\SetKwInOut{input}{input}
\SetKwInOut{output}{output}
\input{$m$ is a mapping and $\tau$ is a GED threshold}
\output{a lower bound of $m$}
\BlankLine
  $dist \leftarrow {\sf ec}(m)$\;
  \lIf{$dist > \tau$}{\Return $dist$}
  \BlankLine
  
  $dist \leftarrow dist + \mathcal{B}(m)$\;
  \lIf{$dist > \tau$}{\Return $dist$}
  \BlankLine

  $e \leftarrow hash.$lookup($V(g_1|_{m})$)\;
  \lIf{$dist + e.lb > \tau$}{\Return $dist + e.lb$}
  \BlankLine

  ${\sf f_{lb}} \leftarrow [{\sf lb_L}, {\sf lb_C}, {\sf lb_P}]$\;
  $i \leftarrow e.index + 1$\;
  \While{$i \leq |{\sf f_{lb}}|$}{
    $e.index \leftarrow i$\;
    \If{$e.lb < {\sf f_{lb}}[i](g'_1, g'_2)$}{
      $e.lb \leftarrow {\sf f_{lb}}[i](g'_1, g'_2)$\;
      \lIf{$dist + e.lb > \tau$}{\Return $dist + e.lb$}}
    $i \leftarrow i + 1$\;}
  \BlankLine

  \Return $dist + e.lb$\;
\caption{${\sf lb_M}(m, \tau)$}
\label{alg:lowerbound}
\end{algorithm}  

\begin{algorithm}[htbp]
\SetKwInOut{input}{input}
\SetKwInOut{output}{output}
\input{$g_1$ and $g_2$ are graphs, and $\tau$ is a GED threshold}
\output{${\sf NassGED}(g_1, g_2)$}
\BlankLine

$queue \leftarrow \emptyset$; $m_r \leftarrow \emptyset$\;
\lIf{${\sf lb_M}(m_r, \tau) \leq \tau$}{ $queue$.push($m_r$)}
\BlankLine

\While{queue $\neq \emptyset$}{
  $m \leftarrow$  $queue$.pop()\;
  \lIf{$|m| = |V(g_2)|$}{ \Return ${\sf lb_M}(m, \tau)$ }

  \ForEach{child node $m_c$ of $m$}{
    \lIf{${\sf lb_M}(m_c, \tau) \leq \tau$}{ $queue$.push($m_c$)}}}
   
 \Return $\tau + 1$\;
\caption{\textsf{NassGED}($g_1$, $g_2$, $\tau$)}
\label{alg:ged}
\end{algorithm}

Algorithm~\ref{alg:ged} encapsulates our GED computation algorithm.  It first
tries to prune the root of the search tree by computing the lower bound of the
root node (Line~2). It is worth to remind that {\sf NassGED} encompasses the
refinement step of {\sf Inves} by applying ${\sf lb_P}$ in
Algorithm~\ref{alg:lowerbound} to the root node.
If it fails to prune the root node, it pushes the root node into the priority
queue, $queue$. Then, it repeatedly expands or prunes the search tree by
investigating currently active tree nodes, which are contained in the queue, as
follows (the {\bf while} loop in Lines~3--7). The algorithm pops a mapping $m$
from the queue that has a minimum lower bound (Line~4). If $m$ is a full 
mapping (i.e., a mapping having all vertices in $g_1$ and $g_2$), it returns
${\sf lb_M}(m)$, which is equal to ${\sf ec}(m)$ since $m$ is a full mapping.
(Line~5). Otherwise, it expands the search tree using each child mapping of the
popped mapping based on the lower bound of the child mapping (Line~7). The
algorithm returns $\tau+1$ if it prunes all possible subtrees of the search
tree (Line~8).

\vspace*{0.7em}
\noindent
{\bf Correctness of Algorithm~\ref{alg:ged}.}  ${\sf lb_M}$ in
Algorithm~\ref{alg:lowerbound} correctly returns a lower bound because (1) each
lower bound function correctly calculates a lower bound~\cite{MIXED, INVES} and
(2) the hash returns a correct lower bound (by
Lemma~\ref{lm:unmappedpart}). Algorithm~\ref{alg:ged} pushes every node of the
search tree whose lower bound is not greater than $\tau$ (Line~7). It returns
if either it finds a full mapping (Line~5) or the queue is empty (Line~8).
Since it pops a mapping having the lowest lower bound from the queue, if the
mapping popped from the queue is a full mapping, it is guaranteed that the
mapping has a minimum edit cost. If the queue is empty, every partial mapping
is pruned since the lower bound is greater than $\tau$, and thus the algorithm
returns $\tau + 1$ to indicate ${\sf NassGED}(g_1, g_2) > \tau$ (Line~8).

\section{Implementation}
\label{sec:impl}

\subsection{Indexing}
\label{sec:index}

Given a graph database $\mathcal{D}$, we need to pre-compute and materialize
the GED between every pair of graphs in $\mathcal{D}$ to obtain $\mathcal{R}(g,
\tau)$ for any graph $g \in \mathcal{D}$ with any distance threshold
$\tau$. However, it is impractical to build such an index. Instead, we assume a
pre-defined maximum threshold $\tau_{max}$. By Lemma~\ref{lm:candgen}, it is
sufficient to compute $\mathcal{R}(g, 2\tau_{max})$ for each graph $g \in
\mathcal{D}$. We use $\tau_{index}$ to denote the maximum GED threshold for
indexing, i.e., $\tau_{index} = 2\tau_{max}$.

\begin{algorithm}[htbp]
\SetKw{KwTo}{to}
\SetKw{KwSync}{sync}
\SetKwInOut{input}{input}
\SetKwInOut{output}{output}
\SetKwProg{Spawn}{spawn}{:}{}
\SetKwFor{RepTimes}{repeat}{times}{}
\BlankLine
\input{$\mathcal{D}$ is a graph database,
  $\tau_{index}$ is a threshold for indexing, $n$ is the number of threads.}
\output{Index $\mathcal{I}$}
$\mathcal{I} \leftarrow$ an array of $|\mathcal{D}|$ empty lists\;
\RepTimes{n}{  
\Spawn{}{
  $i \leftarrow {\sf next\_graph\_id()}$; \tcp*[f]{\tt synchronous access}\\
  \lIf{$i \geq |\mathcal{D}|$}{\Return}
  \BlankLine
  $\mathcal{I}[g_i]$.{\sf append}$(g_i, 0)$\;
  \For{$j \leftarrow i+1$ \KwTo $|\mathcal{D}|$}{
    $\delta \leftarrow {\sf NassGED}(g_i, g_j, \tau_{index})$\;
  
    $\mathcal{I}[g_i]$.{\sf append}$(g_j, \delta)$\;
    $\mathcal{I}[g_j]$.{\sf append}$(g_i, \delta)$\;
  }
}
}
\KwSync\;
\Return $\mathcal{I}$\;
\caption{\sf \textsf{NassIndex}($\mathcal{D}$, $\tau_{index}, n$)}
\label{alg:nassindex}
\end{algorithm}

Since we need $|\mathcal{D}|$ independent similarity searches to build an
index, we implement a straightforward multi-threading to reduce index building
time. Our implementation is to spread each data graph to a different thread,
and perform similarity searches simultaneously\footnote{There can be
  alternative implementations, e.g., improving GED computation using
  multi-threads or improving a similarity search by spreading candidates to
  different threads.  However, parallel graph search is out of the scope of the
  paper and we will leave this as future work.}. Algorithm~\ref{alg:nassindex}
shows our indexing algorithm. After initializing the index (Line~1), it spawns
$n$ threads (the loop in Line~2). Each thread synchronously gets a graph id $i$
(Line~4), which used to indicate the $i^{th}$ graph $g_i$ in $\mathcal{D}$, and
computes the GED between $g_i$ and $g_j$ for $j > i$ (Lines~7--8). Then,
it updates the index entries $\mathcal{I}[g_i]$ and $\mathcal{I}[g_j]$ with the
GED (Lines~9--10). 

One problem in indexing is that GED computation with $2\tau_{max}$ can be too
costly to be practical. We solve the problems by restricting a maximum
threshold for an index to $\tau_{index} = \tau_{max} + c$, where $c$ is a
constant less than $\tau_{max}$. For a query graph $q$, if we find a data graph
$g$ such that ${\sf ged}(q, g) \leq c$, by Lemma~\ref{lm:candgen}, we can
(re)generate candidate graphs using the index for all possible thresholds $1
\leq \tau \leq \tau_{max}$. In Table~\ref{tbl:expr-intro-cand}, for example,
the average number of results is greater than 1 when $\tau = 3$. By using $c =
3$ for this dataset, therefore, we can expect almost all queries can take
advantage of our index.

Another problem is that GED computation of a certain pair of graphs can be
intractable even with a reasonably large threshold. We solve the problem by
allowing an inexact index entry having a GED lower bound for such a pair of
graphs. To this end, we assume that the time consumption is proportional to the
memory consumption in computing a GED, and we maintain a thread that monitors
the real memory consumption of the indexing process. If the memory consumption
reaches a pre-defined limit, we select a victim thread that has the largest
queue size of {\sf NassGED}. The victim thread immediately returns the minimum
lower bound among the lower bounds of queued nodes.

\begin{algorithm}[htbp]
\SetKwInOut{input}{input}
\SetKwInOut{output}{output}
\BlankLine
  \If(\tcp*[f]{regenerate a candidate set}){$\tau + \delta \leq \tau_{index}$}{
    $\mathcal{A} \leftarrow \{r\ |\ (r, d) \in \mathcal{I}[g]\ \land\ d \leq \tau - \delta\ \land\ d {\tt\ is\ exact} \}$\;
    $\mathcal{R}_g \leftarrow \{r\ |\ (r, d) \in \mathcal{I}[g]\ \land\ d \leq \tau + \delta\}$\;
  
    $\mathcal{C}' \leftarrow (\mathcal{C} - \mathcal{V}) \cap
    (\mathcal{R}_g - \mathcal{A})$\;}
  \Else(\tcp*[f]{keep verifying the current candidate set}){
    $\mathcal{A} \leftarrow \{g\}$\;
    $\mathcal{C}' \leftarrow (\mathcal{C} - \mathcal{V})$\;}
 \caption{\sf Replacement of Lines~7--8 of Algorithm~\ref{alg:nass}}
\label{alg:nass-modified}
\end{algorithm}

To use our index, we modify Lines~7--8 of Algorithm~\ref{alg:nass} as shown in
Algorithm~\ref{alg:nass-modified}. According to Lemma~\ref{lm:candgen}, we
regenerate candidates only when $\delta + \tau \leq \tau_{index}$ (Line~1).
For $\mathcal{R}(g, \tau - \delta)$ in Line~7 of Algorithm~\ref{alg:nass}, we
include only those graphs having exact GEDs (Line~2). For $\mathcal{R}(g, \tau
+ \delta)$ in Line~8 of Algorithm~\ref{alg:nass}, which is used to regenerate
candidates, we use an approximate set that includes inexact GEDs
($\mathcal{R}_g$ in Line~3).

\vspace*{0.7em}
\noindent
{\bf Correctness of Algorithm~\ref{alg:nass-modified}.} Since an inexact GED is
a GED lower bound (i.e., $\mathcal{R}(g, \tau+\delta) \subseteq \mathcal{R}_g$)
and uncollected results in Line~2 are included in $\mathcal{C}'$ (Line 4),
$\mathcal{C}'$ contains all remaining results. Therefore, the algorithm does
not miss any result in spite of inexact index entries.
If the algorithm cannot use the index (Line~5), it continues to verify the
current the candidate set (Line~7). Therefore, a restricted $\tau_{index}$ does
not affect the result set of the query.

\subsection{GED Computation}
\label{sec:gedimpl}

As discussed in {\sf Inves}~\cite{INVES}, a proper vertex ordering of $g_2$ is
crucial to the performance of GED computation. In this paper, we abide by the
vertex ordering of ${\sf Inves}$. Because we apply its partitioning technique
to the root node of the search tree, we can immediately use the vertex ordering
obtained from the partitioning of $g_2$ (refer to \cite{INVES} for the
details).

We use a balanced binary search tree to implement the hash for storing ${\sf
  f_{lb}}(g_1 \backslash g_1|_m, g_2 \backslash g_2|_m)$. It can be easily seen
that the time complexity for the hash is exactly the same with that for the
priority queue used in the GED computation algorithm. A bitmap, which is used
as the key for the hash, is created for each mapping $m$ to represent
$V(g_1|_m)$ as follows. The bitmap has $|V(g_1)|$ bits, and the $i^{th}$ bit is
1 if the $i^{th}$ vertex of $g_1$ is included in $V(g_1|_m)$, otherwise the
$i^{th}$ bit is 0.
Apparently, the bitmap of $m$ is incrementally created using the bitmap of the
parent of $m$ (i.e., the mapping in the parent node of $m$ in the search tree)
by setting one bit for the vertex newly inserted to $g_1|_{m}$. Since we focus
on small and medium sized graphs, one or two 64-bit integers are sufficient for
a bitmap in most cases.

The lower bound of a mapping can be incrementally computed using its parent
mapping. Consider a partial mapping $m_p$ and its child $m_c = m_p \cup \{u'
\smapsto v'\}$. We compute lower bounds as follows\footnote{We note that
  $\lambda$, which is introduced in Section~2, is not a label. It is used for
  indicating the absence of an edge or a vertex, and thus, we do not include
  $\lambda$ in any label multiset discussed in this paper.}.

\vspace*{0.2cm}
\noindent
{\bf Bridge cost.} The label multisets of bridges of $u'$ and $v'$ are constructed from the scratch.
The label multisets of bridges of other vertices are updated as:    
\vspace*{-0.1cm}
  \[\forall u \in V(g_1|_{m_p})\ \ L_{br}^{m_c}(u) \leftarrow L_{br}^{m_p}(u) - \{l(u, u')\},\]
\vspace*{-0.4cm}
  \[\forall v \in V(g_2|_{m_p})\ \ L_{br}^{m_c}(v) \leftarrow L_{br}^{m_p}(v) - \{l(v, v')\}.\]
Then, $\mathcal{B}(m_c)$ is computed using the label multisets of bridges constructed for $m_c$.

\vspace*{0.2cm}
\noindent  
{\bf Label-based lower bound.} Let $g'_1$, $g'_2$, $g''_1$, and $g''_2$ denote
$g_1 \backslash g_1|_{m_p}$, $g_2 \backslash g_2|_{m_p}$, $g_1 \backslash
g_1|_{m_c}$, and $g_2 \backslash g_2|_{m_c}$, respectively.
The label multisets of unmapped subgraphs are incrementally constructed as
follows.
\vspace*{-0.1cm}
\[L_V(g''_1) \leftarrow L_V(g'_1) - \{ l(u')\},\ L_V(g''_2) \leftarrow L_V(g'_2) - \{ l(v')\}, \phantom{.}\]
\vspace*{-0.4cm}
\[L_E(g''_1) \leftarrow L_E(g'_1) - \{l(u, u')\ |\ u \in V(g'_1) \land (u, u') \in E(g'_1)\},\]
\vspace*{-0.4cm}
\[L_E(g''_2) \leftarrow L_E(g'_2) - \{l(v, v')\ |\ v \in V(g'_2) \land (v, v') \in E(g'_2)\}.\]
Then, ${\sf lb_L}$ is computed using the label multisets of $g''_1$ and
$g''_2$.

\vspace*{0.2cm}
\noindent  
{\bf Compact branch-based lower bound.} We remove $b_{u'}$ and $b_{v'}$ from
$B(g'_1)$ and $B(g'_2)$, respectively, and compute ${\sf lb_C}$ again. As we
mentioned earlier, we use $O(n\log n)$ algorithm for finding a minimum weighted
bipartite matching between $B(g'_1)$ and $B(g'_2)$, where $n = |B(g'_1)|$
(refer to \cite{MIXED} for the details of the algorithm). The algorithm
basically merges $B(g'_1)$ and $B(g'_2)$ after sorting the branch
sets. Therefore, the time complexity of ${\sf lb_C}$ is dominated by the cost
for sorting $B(g'_1)$ and $B(g'_1)$. As we remove $b_{u'}$ and $b_{v'}$ from
the already sorted branch sets, our incremental implementation requires $O(n)$,
which is the cost for merging the sets.

\vspace*{0.2cm}
\noindent  
{\bf Partition-based lower bound.}  Unlike other lower bounds, it is not
straightforward to incrementally compute the lower bound ${\sf lb_P}$ from
parent's ${\sf lb_p}$, and thus we do not use parent's ${\sf lb_p}$. Instead,
we take a different approach to save computation for ${\sf lb_p}$. Consider
$\tau'_1 = \tau - ({\sf ec}(m_1) + \mathcal{B}(m_1))$ for a partial mapping
$m_1$. As the partitioning technique of {\sf Inves} incrementally increases the
lower bound, we can save computation by stopping partitioning as soon as ${\sf
  lb_P} = \tau'_1 + 1$. Consider $\tau'_2 = \tau - ({\sf ec}(m_2) +
\mathcal{B}(m_2))$ for another partial mapping $m_2$ such that the unmapped
subgraphs of $m_1$ and $m_2$ are the same. If $\tau'_2 > \tau'_1$, we cannot
prune $m_2$ using the ${\sf lb_p}$ of $m_1$ stored in the hash. Since ${\sf
  lb_p}$ may not be tight, we can compute ${\sf lb_p}$ again for $m_2$. If the
stored ${\sf lb_p}$ is not tight, we can resume partitioning to tighten ${\sf
  lb_p}$ instead of computing ${\sf lb_p}$ from the scratch. It is
straightforward to stop and resume partitioning and we omit the details (refer
to \cite{INVES} for the details of the partitioning technique of {\sf Inves}).

\vspace*{0.2cm}
To save the memory, instead of keeping bridge multisets, label multisets, and
branch multisets in each mapping, we compute multisets for the parent $m_p$
popped from the queue, and use them to incrementally compute label multisets of
each child of $m_p$.

\section{Experiments}
\label{sec:expr}

\subsection{Experimental Setting}
\label{sec:expr-setting}

We conducted experiments on two widely used datasets, AIDS and PubChem. AIDS is
an antiviral screen compound data set published by NCI/NIH\footnote{{\sf
    https://cactus.nci.nih.gov/download/nci/AIDS2DA99.sdz}}. It is a popular
benchmark used in most graph search techniques. PubChem is a chemical compound
dataset\footnote{{\sf https://pubchem.ncbi.nlm.nih.gov,
    Compound\_000975001\_001000000.sdf}}. It is a subset of chemical compounds
published by the PubChem Project. Graphs in the PubChem dataset contain
repeating substructures and have less size and label variations compared with
the AIDS dataset. Table~\ref{tbl:stats} shows statistics of the datasets. In
the table, $|\mathcal{D}|$ is the number of graphs in each dataset,
$|V|_{\text{avg}}$ and $|E|_{\text{avg}}$ is the average numbers of vertices
and edges, $\sigma_{|V|}$ and $\sigma_{|E|}$ are the standard deviations of the
numbers of vertices and edges, and $n_{vl}$ and $n_{el}$ are the numbers of
distinct vertex and edge labels.

\newcolumntype{R}[1]{>{\raggedleft\arraybackslash}p{#1}}

\begin{table}[htbp]
  \centering
  \caption{Statistics of datasets}
  \begin{tabular}{l|R{0.8cm}|R{0.7cm}|R{0.7cm}|R{0.45cm}|R{0.45cm}|R{0.37cm}|R{0.37cm}} \thickhline
    Dataset& $|\mathcal{D}|$ & $|V|_{\text{avg}}$ & $|E|_{\text{avg}}$ & $\sigma_{|V|}$ & $\sigma_{|E|}$ & $n_{vl}$ & $n_{el}$\\ \hline\hline
    AIDS & 42,689 & 25.60 & 27.60 & 12.2 & 13.3 & 62 & 3\\
    PubChem & 22,794 & 48.11 & 50.56 & 9.4 & 9.9 & 10 & 3\\ \thickhline
  \end{tabular}
  \label{tbl:stats}
\end{table}

\noindent
We also used synthetic datasets to evaluate the scalability of {\sf Nass} (see
Section~\ref{sec:scale} for details).

We randomly sampled 100 query graphs from each dataset. If we find a data graph
which is the same as a query graph, our index can immediately find all results of the
query by Lemma~\ref{lm:candgen} and Lemma~\ref{lm:regen}. Thus, we removed the
query graphs from the dataset not to exaggerate the performance gain of {\sf
  Nass}. Aggregated results of 100 queries are reported in the experiments.
We note that y-axis is log-scaled in all experiments.
For experiments on the AIDS and PubChem datasets, we set the maximum threshold
$\tau_{max} = 6$.

We implemented {\sf Nass} in C++, and compiled it using GCC with the -{\sf O3}
flag\footnote{The source code of {\sf Nass} is available at {\sf
    https://github.com/JongikKim/Nass}.}. We compared {\sf Nass} with two
representative indexing techniques: {\sf Pars}~\cite{PARS, PARS2}, and {\sf
  MLIndex}~\cite{MLINDEX}, and two state-of-the art GED verification
techniques: {\sf Inves}\footnote{The source code of {\sf Inves} is obtained
  from {\sf https://github.com/JongikKim/Inves}.}~\cite{INVES} and {\sf
  CSI\_GED}\footnote{The binary code of {\sf CSI\_GED} is obtained from the
  authors.}~\cite{CSIGED, CSIGED2}. Since the indexing techniques, {\sf Pars}
and {\sf MLIndex}, mainly rely on the out-dated {\sf A*-GED} for verification,
we used {\sf Inves} in the verification phase of them, similar to~\cite{INVES}.
\COMMENTOUT{ We use ${\sf N}_k$ and ${\sf Nass}_k$ to denote {\sf Nass} with an
  index whose $\tau_{index}$ is $k$. We also use ${\sf N}_k$ to denote the
  index itself if clear from the context.  } All experiments were conducted on
a machine with 32GB RAM, and an Intel core i7, running a 64-bit Ubuntu OS. Data
graphs and indices were kept in memory.

\subsection{Experiments on Index}
We constructed indices by varying the limit of memory consumption for building
an index from 1GB to 8GB.  On the AIDS dataset, we used $\tau_{index} =
\tau_{max} + 3$ based on the observation that the average number of results of
queries is 1.26 when $\tau = 3$. As described in Section~\ref{sec:index}, almost
all queries can take advantage of an index built with the $\tau_{index}$. On the
PubChem dataset, we used $\tau_{index} = \tau_{max} + 1$ since the average
number of results of queries is 0.94 when $\tau = 1$. We used 8 threads to
construct an index.

Experimental results for indexing is shown in Table~\ref{tbl:index}. In the
table, {\sf T}, {\sf I}, and {\sf N} denote the index construction time, the
percentage of inexact entries, and the number of indexed entries, respectively.

\begin{table}[h!]
  \centering
  \caption{Experimental results on indexing}
  \small
  \begin{tabular}{c|l|r|r|r|r} \thickhline
    &  \multirow{2}{*}{Dataset} & \multicolumn{4}{c}{Memory limit for indexing}\\\cline{3-6}
    &         & 1GB & 2GB & 4GB & 8GB \\ \hline\hline
    \multirow{2}{*}{\sf T}
              & AIDS     & 190379s & 192211s & 193319s & 193267s \\
              & PubChem  & 83520s & 114323s  & 152475s & 192978s \\ \thickhline
    \multirow{2}{*}{\sf I}
              & AIDS     &  0.0033\% & 0.0011\% & 0.0006\% & 0.0002\%\\
              & PubChem  &  5.18\% & 2.74\% & 1.72\% & 0.85\% \\ \thickhline
    \multirow{2}{*}{\sf N}
              & AIDS     & 4220628 & 4220658 & 4220606 & 4220588 \\
              & PubChem  & 105414 & 103596 & 102992 & 102302 \\ \thickhline
  \end{tabular}
  \label{tbl:index}
\end{table}

A similarity search on AIDS and PubChem datasets typically used about 300MB
memory in our implementation. However, the search space for GED computation of
a certain pair of graphs sharply increased, and required a tremendous amount of
memory. By limiting the memory consumption, therefore, we restricted the time
for computing GED between such a pair of graphs.
Although it takes much time to construct an index, we remark that an index
is pre-built offline and many online queries with different thresholds can take
advantage of the index (c.f., {\sf Pars}~\cite{PARS, PARS2} also spends more
than $10^5$ seconds to build an index for the AIDS dataset).
On the AIDS dataset, the percentage of inexact index entries was negligibly
small. On the PubChem dataset, it was from 1\% to 5\% only.
The index size can be measured by the number of indexed entries, where
each entry contains a graph id with a GED between 0 and $\tau_{index}$.  If we
use a compact representation, an entry requires $\lceil \log_2 |\mathcal{D}|
\rceil + \lceil \log_2 (\tau_{index}+1) \rceil + 1$ bits, where 1 bit is used
to indicate the exactness of the GED. The size of an index can be calculated by
multiplying the number of indexed entries by the size of an entry.
For example, the PubChem index with 4GB limit requires $102992 \times
\lceil(\lceil \log_2 22774 \rceil + \lceil \log_2 8 \rceil + 1)/8\rceil\ {\rm bytes}
\approx 309 {\rm KB}$. Similarly, the AIDS index with 4GB limits requires about 12MB.

\begin{figure}[htbp]
  \centering
  \subfigure[Query time (AIDS)]{\includegraphics[height=2.2cm]{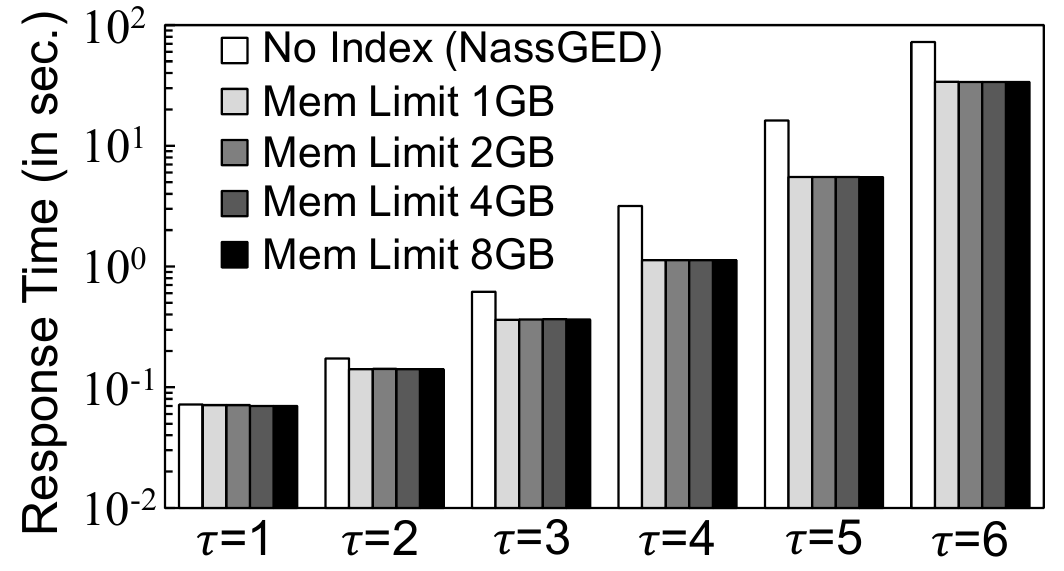}}
  \subfigure[Query time (PubChem)]{\includegraphics[height=2.2cm]{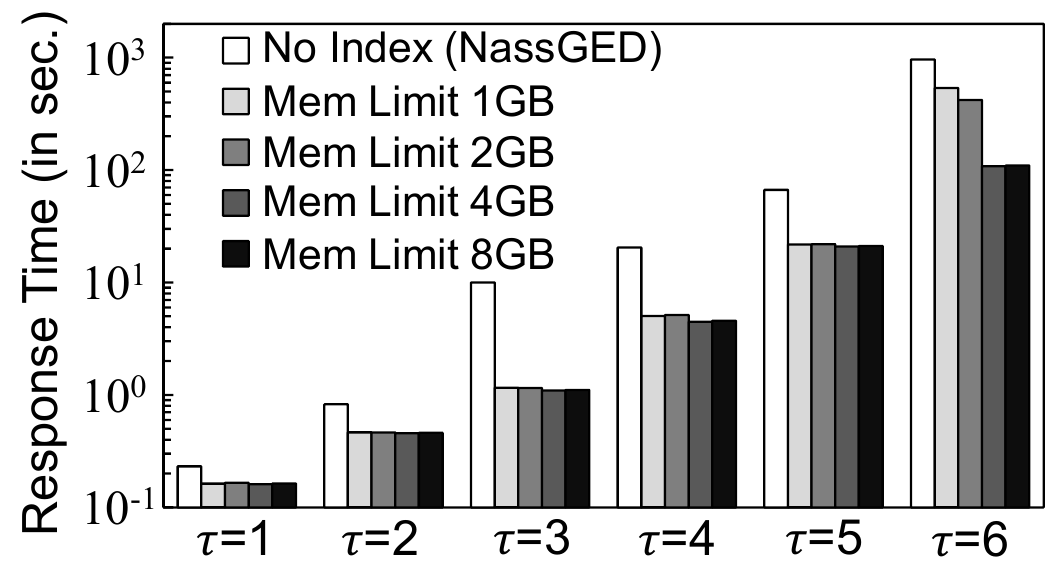}}
  \caption{Query response time for difference indices}
  \label{fig:expr-index}
\end{figure}

Figure~\ref{fig:expr-index} shows the query response time for indices
constructed with different memory limits. In the figure, {\sf No Index} denotes
{\sf Nass} without an index (i.e., each graph was directly verified through
{\sf NassGED}). On the AIDS dataset, a similarity search with an index was up
to 3 times faster than that without an index as shown in
Figure~\ref{fig:expr-index}(a). On the PubChem dataset, similarly, an indexed
search was up to 9 times faster than a search without an index as shown in
Figure~\ref{fig:expr-index}(b).
Based on the experiments, we chose the indices with the 4GB memory limit for
both datasets. The reported results in the following section are based on the
indices.

\subsection{Evaluating Graph Similarity Search}
\label{sec:expr-simsearch}

Figure~\ref{fig:expr-search} shows the query evaluation results of different
search algorithms. In the figure, we use {\sf P}, {\sf M}, and {\sf N} to
denote {\sf Pars}+{\sf Inves}, {\sf MLIndex}+{\sf Inves}, and {\sf Nass}
respectively. {\sf Nass} consistently outperformed all existing algorithms for
all thresholds as shown in Figure~\ref{fig:expr-search}(a) and (b).
${\sf Nass}$ was from 4 to 13 times faster than existing algorithms on the AIDS
dataset (Figure~\ref{fig:expr-search}(a)), and from 4 to 60 times faster on the
PubChem dataset (Figure~\ref{fig:expr-search}(b)).

\begin{figure}[htbp]
  \centering
  \subfigure[Query time (AIDS)]{\includegraphics[height=2.36cm]{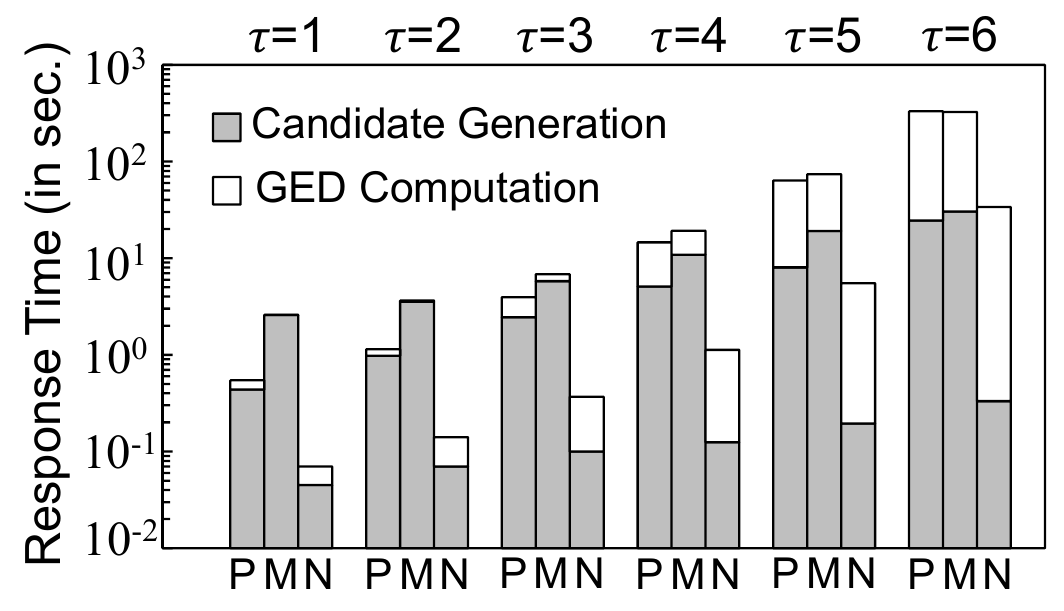}}
  \subfigure[Query time (PubChem)]{\includegraphics[height=2.36cm]{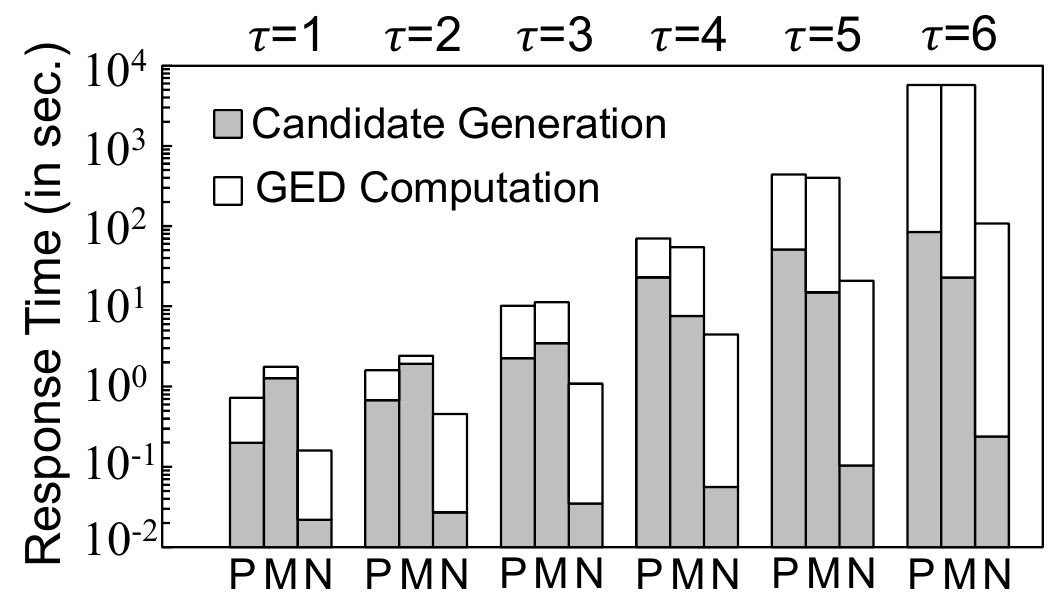}}
  \subfigure[Candidates (AIDS)]{\includegraphics[height=2.2cm]{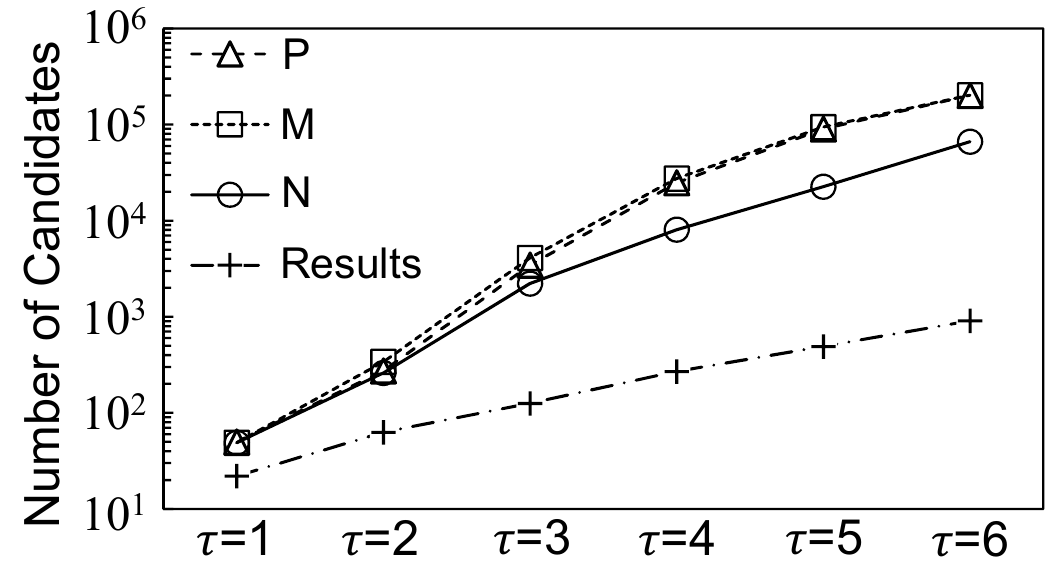}}
  \subfigure[Candidates (PubChem)]{\includegraphics[height=2.2cm]{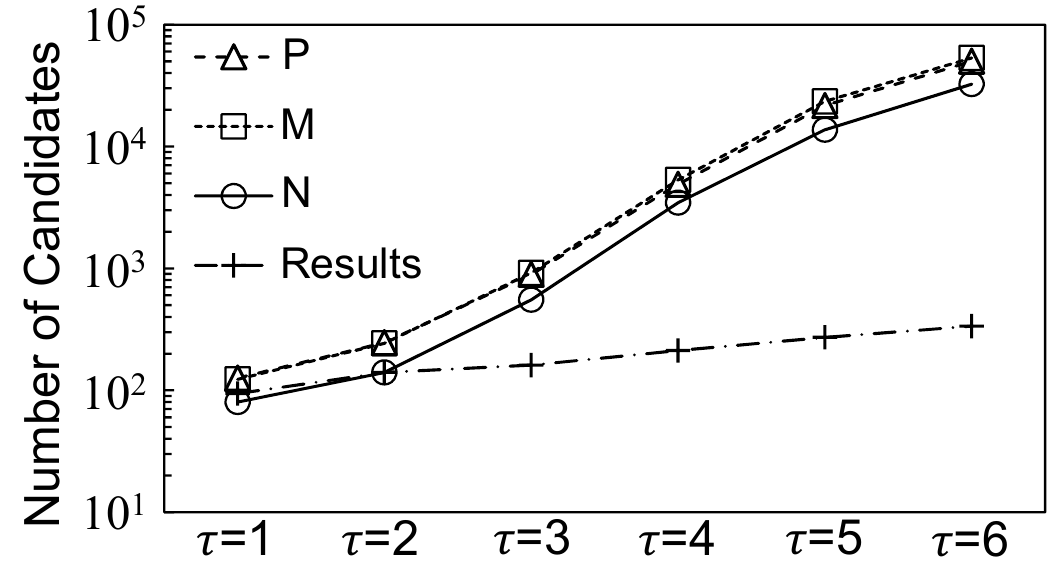}}
  \subfigure[Mappings (AIDS)]{\includegraphics[height=2.2cm]{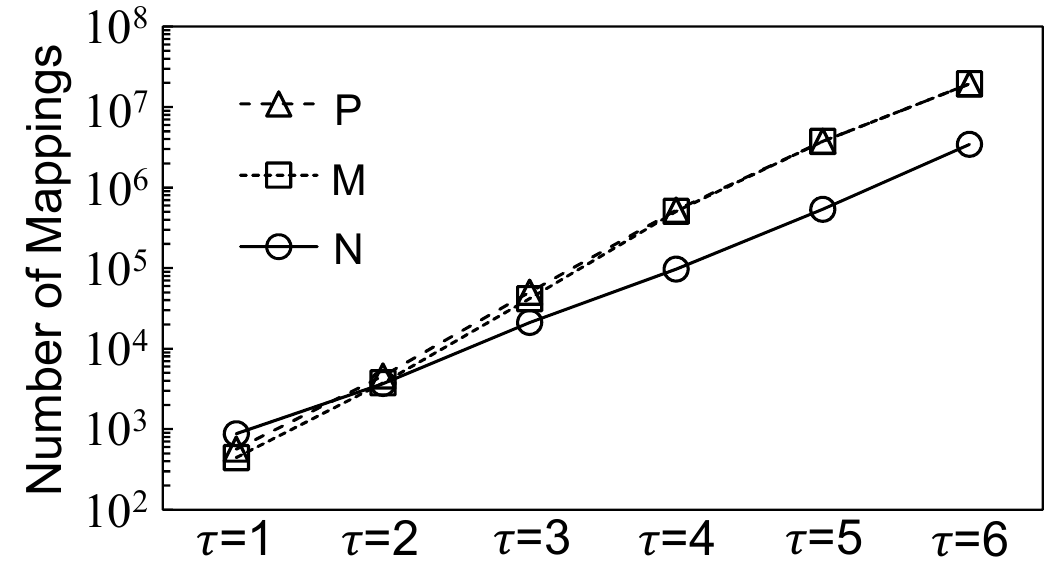}}
  \subfigure[Mappings (PubChem)]{\includegraphics[height=2.2cm]{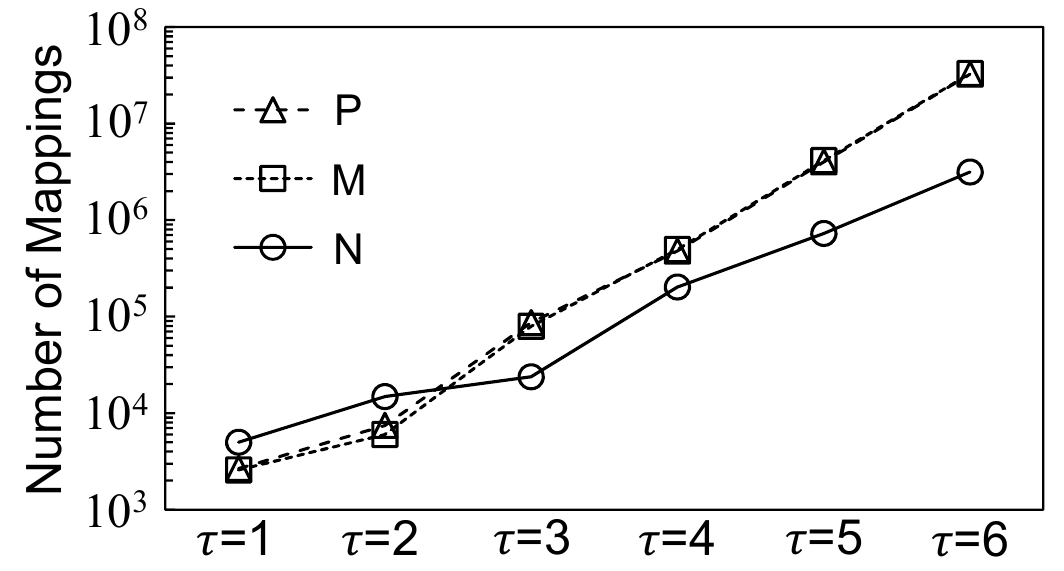}}
  \caption{Comparing {\sf Nass} with existing search techniques}
  \label{fig:expr-search}
\end{figure}

The improvement of {\sf Nass} can be explained by the number of
candidates\footnote{The number of candidates of {\sf Nass} was counted using
  those candidates that survived from the filtering pipeline in the root node
  of the search tree (see Section~\ref{sec:gedmodel} and
  Section~\ref{sec:GEDalgorithm} for the reason).} that require GED computation
in Figure~\ref{fig:expr-search}(c) and (d), and the number of mappings pushed
into the queue while GED computation in Figure~\ref{fig:expr-search}(e) and
(f). The number of candidates generated by {\sf Nass} was up to 4 times smaller
than that of existing candidates. Interestingly, on the PubChem dataset, the
number of candidates generated by {\sf Nass} was fewer than the number of
result graphs when $\tau = 1$ (Figure~\ref{fig:expr-search}(d)). This is
because {\sf Nass} can identify some result graphs without verification.
Since our index significantly reduced the number
of candidates and our GED algorithm effectively prunes the search tree, the
number of mappings pushed into the queue was dramatically reduced as depicted
in Figure~\ref{fig:expr-search}(e) and (f). When $\tau = 5$ on the AIDS
dataset, for example, the total number of mappings pushed into the queue by
${\sf Nass}$ was about 10 times smaller than that of existing techniques.
For a low threshold, however, the number of mappings of {\sf Nass} was slightly
greater than that of existing techniques (e.g. $\tau = 1$ on the AIDS
dataset). This is because of the {\it partial GED} function of the {\sf Inves}
verification technique (see {\sf Inves}\cite{INVES} for the details). We cannot
apply the partial GED function in our GED computation algorithm, because {\sf
  Nass} requires an exact distance for a result graph returned by {\sf NassGED}
but the partial GED function returns an inexact distance for a
result. Nonetheless, {\sf Nass} was much faster than the existing techniques on
low thresholds, because the candidate generation of {\sf Nass} was extremely
efficient compared with that of existing indexing techniques.

\subsection{Evaluating GED Verification}
\label{sec:expr-ged}
We compared our GED computation algorithm, denoted by {\sf NassGED}, with {\sf
  Inves} and {\sf CSI\_GED}. We evaluated the performance of GED computation as
follows. For each query, we first applied the label filtering to every graph in
a dataset, and then directly verified each graph that passed the label filter.

\begin{figure}[htbp]
  \centering
  \subfigure[Query time (AIDS)]{\includegraphics[height=2.2cm]{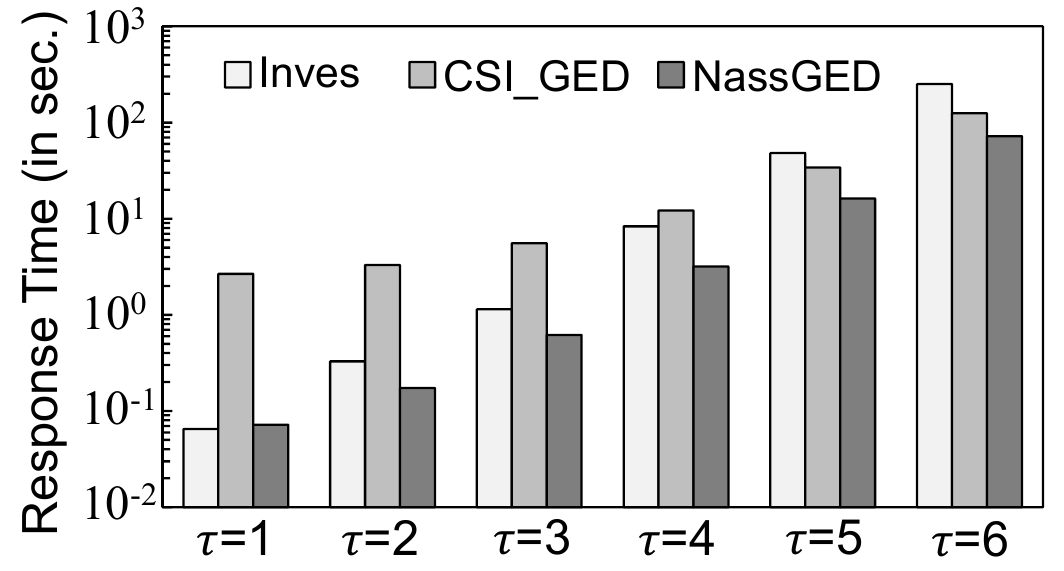}}
  \subfigure[Query time (PubChem)]{\includegraphics[height=2.2cm]{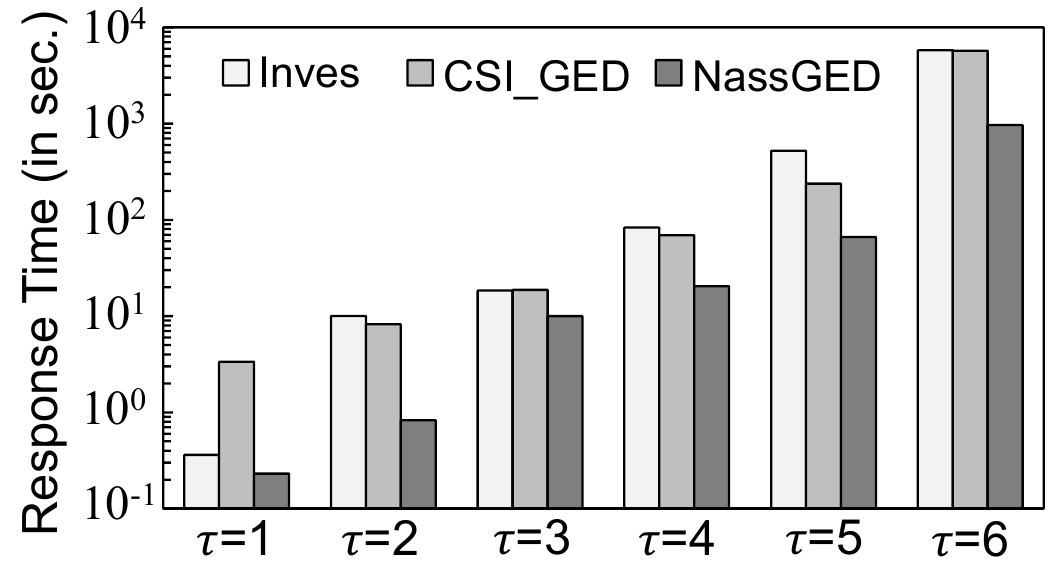}}
  \caption{Comparing {\sf NassGED} with existing GED computation algorithms}
  \label{fig:expr-ged}
\end{figure}

Figure~\ref{fig:expr-ged}(a) and (b) show the results on the AIDS and PubChem
datasets, respectively. {\sf NassGED} consistently outperformed {\sf Inves} and
{\sf CSI\_GED} for $\tau \geq 2$. On the AIDS dataset, {\sf Inves} slightly
outperformed {\sf NassGED} when $\tau = 1$, but the difference was negligible.
On lower thresholds, {\sf NassGED} and {\sf Inves} performed much better than
{\sf CSI\_GED}. For $\tau \geq 3$, {\sf NassGED} was up to 2.5 times faster
than existing algorithms on the AIDS, and up to 6 times faster on PubChem
datasets.

\begin{figure}[htbp]
  \centering
  \subfigure[Query time (AIDS)]{\includegraphics[height=2.2cm]{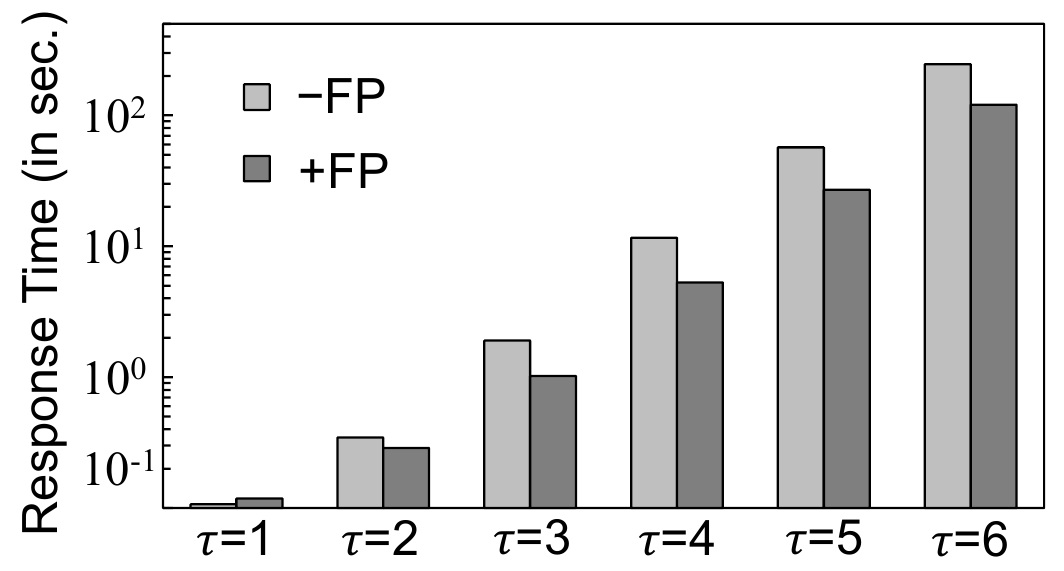}}\hfill
  \subfigure[Query time (PubChem)]{\includegraphics[height=2.2cm]{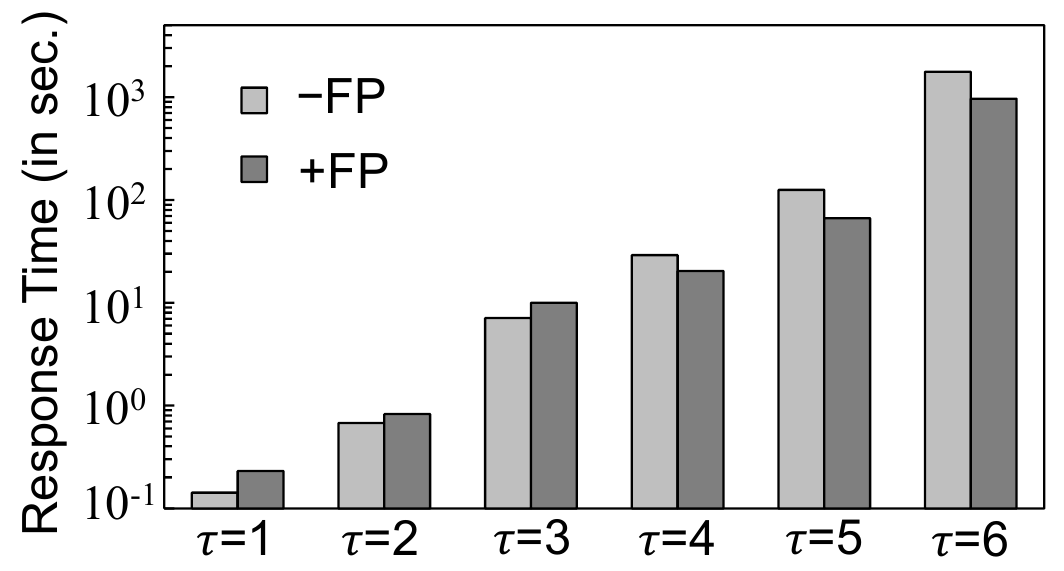}}
  \subfigure[Mappings (AIDS)]{\includegraphics[height=2.2cm]{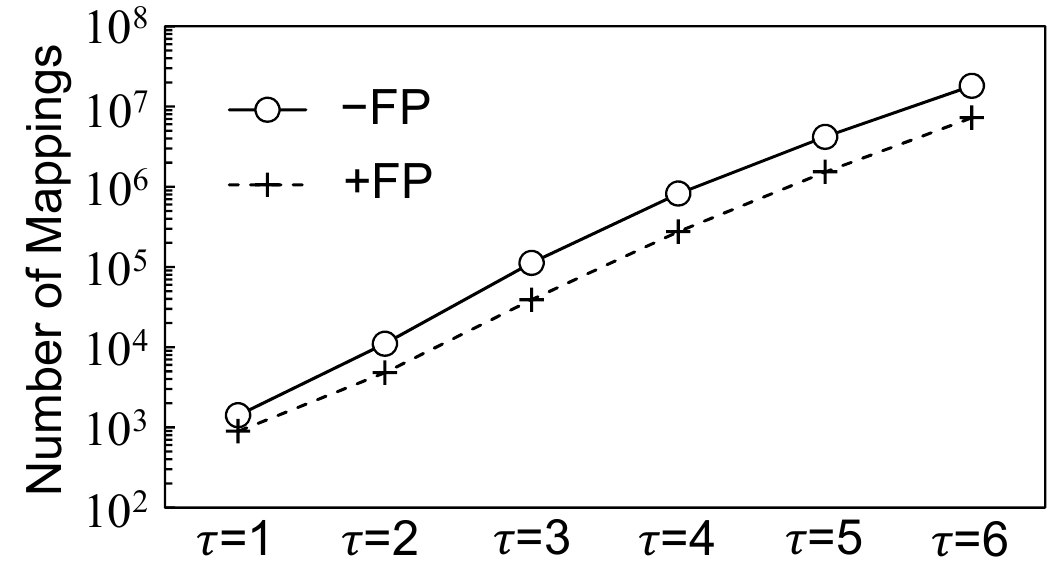}}\hfill
  \subfigure[Mappings (PubChem)]{\includegraphics[height=2.2cm]{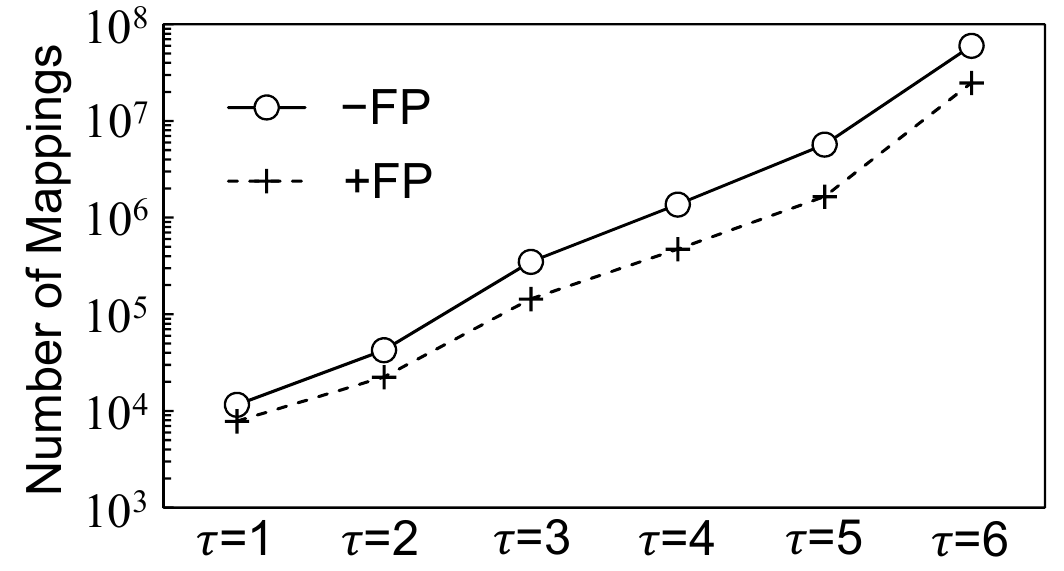}}
  \caption{Evaluating the filtering pipeline of {\sf NassGED}}
  \label{fig:expr-fp}
\end{figure}

The improvement of {\sf NassGED} can be explained by the filtering pipeline
(Section~\ref{sec:filtering}) and efficient implementation
(Section~\ref{sec:gedimpl}). Figure~\ref{fig:expr-fp}(a) and (b) show the
effect of the filtering pipeline on AIDS and PubChem. In the figures, {\sf +FP}
and {\sf -FP} denote {\sf NassGED} with and without the filtering pipeline,
respectively. As shown in the figure, {\sf +FP} improved GED computation by up
to 2.2 times. {\sf +FP} was slightly slower than {\sf -FP} for low thresholds
(e.g., $\tau = 1$ for AIDS and $\tau \leq 3$ for PubChem), because of the
overhead of lower bound computation. Although {\sf +FP} requires more
computation for lower bounds, we observed that the overhead of lower
bound computation did not affect the performance significantly.
The total number of mappings pushed into the queue while GED computation was
shown in Figure~\ref{fig:expr-fp}. {\sf +FP} reduced the number of mappings by
up to 3.3 times. 

\subsection{Scalability Test}
\label{sec:scale}

In this subsection, we report the results of the scalability of {\sf Nass}.
For the experiments, we generated synthetic datasets using a graph generator
{\sf GraphGen}\footnote{\sf https://www.cse.ust.hk/graphgen}. The
generator measures the graph size in terms of the number of edges ($|E|$), and
the density of a graph as $\frac{2|E|}{|V|(|V|-1)}$. We set up the generator as
follows: the average size of graphs is 40; the numbers of distinct vertex and
edge labels are 5 and 2, respectively; and the density of each graph is 0.2. We
initially generated 4k, 8k, 12k, 16k, and 20k datasets. For each graph in a
dataset, we generated 4 more graphs by randomly applying 2, 4, 6, 8, or 10 edit
operations to the graph 4 times. For scalability test, we used $\tau_{max} = 7$
and $\tau_{index} = \tau_{max} + 1$. Figure~\ref{fig:expr-scale} shows the
results. For various thresholds, {\sf Nass} scaled well to large datasets as
shown in the figure.

\begin{figure}[htbp]
  \centering{\includegraphics[height=1.9cm]{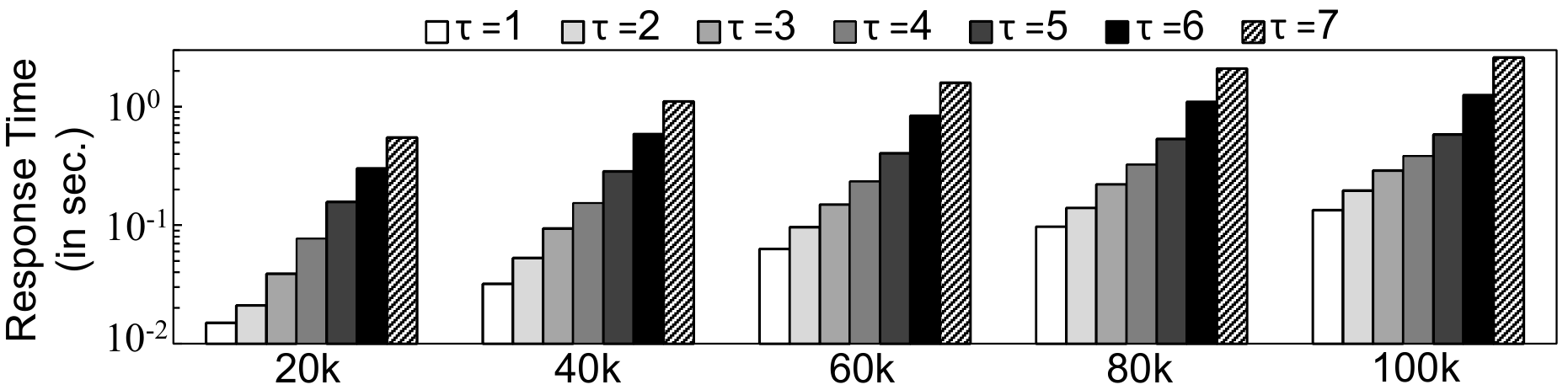}}
  \caption{Evaluating the scalability of {\sf Nass}}
  \label{fig:expr-scale}
\end{figure}

\section{Conclusions}
\label{sec:conclusion}

In this paper, we proposed a completely different approach to graph similarity
search. We generate candidate graphs via GED verification and verify each
candidate via various filtering techniques. The proposed search framework {\sf
  Nass} substantially reduces the number of candidates by dynamically
regenerating candidates while verifying candidates. To efficiently verify each
candidate, our GED computation algorithm utilizes various filtering techniques
to significantly prune the search space of the prefix tree.
We conducted extensive experiments on both real and synthetic
datasets, and the results showed that {\sf Nass} outperformed the state-of-the
art algorithms by an order of magnitude.

\balance

\bibliographystyle{ACM-Reference-Format}
\bibliography{ref}

\end{document}